\documentclass[11pt,a4paper]{article}
\usepackage[authoryear,round]{natbib}
\usepackage{url}
\usepackage{amsmath,amsthm,amssymb}
\usepackage{geometry}
\usepackage{microtype}
\usepackage{booktabs}
\usepackage{hyperref}
\usepackage{colortbl}
\usepackage{threeparttable}
\usepackage{array}
\usepackage{adjustbox}
\usepackage{multirow}
\usepackage{xcolor}
\usepackage{float}
\usepackage{xurl} 
\usepackage{placeins}
\usepackage{booktabs}
\usepackage[table]{xcolor}
\usepackage{makecell}
\usepackage{array}

\newcolumntype{L}[1]{>{\raggedright\arraybackslash}p{#1}}
\newcolumntype{C}[1]{>{\centering\arraybackslash}p{#1}}
\newcolumntype{R}[1]{>{\raggedleft\arraybackslash}p{#1}}
\usepackage{adjustbox}
\usepackage{siunitx}
\usepackage[dvipsnames,table]{xcolor}
\hypersetup{colorlinks=true,linkcolor=blue,citecolor=blue,urlcolor=blue}
\usepackage[linesnumbered, ruled, vlined]{algorithm2e}
\newtheorem{theorem}{Theorem}[section]
\newtheorem{lemma}[theorem]{Lemma}
\newtheorem{proposition}[theorem]{Proposition}

\newtheorem{definition}[theorem]{Definition}
\theoremstyle{remark}
\newtheorem{remark}[theorem]{Remark}

\newcommand{\kap}{\kappa}

\newcommand{\ctv}{c_{\mathrm{tv}}}
\newcommand{\cseed}{c_{\mathrm{seed}}}
\newcommand{\vP}{v_{P1}}
\newcommand{\vPt}{v_{P3}}
\newcommand{\eps}{\varepsilon}
\newcommand{\kk}{\kappa}
\newcommand{\aone}{a_1}
\newcommand{\athree}{a_3}

\newcommand{\Dden}{\mathcal{D}}

\hypersetup{colorlinks=true,linkcolor=blue!60!black}

\begin{document}

\title{A Fast Implied Volatility Method with Expansions}
 \author{Alper Hekimoglu  \thanks{European Investment Bank, 2950 Luxembourg, Luxembourg.~\href{mailto:a.hekimoglu@eib.org}{a.hekimoglu@eib.org}} and Ismail Hakki Gokgoz  \thanks{Central Bank of UAE, 854, Abu Dhabi, UAE (United Arab Emirates).~\href{mailto:ihgokgoz@gmail.com}{ihgokgoz@gmail.com}}}
\date{\today}
\maketitle

\begin{abstract}
We propose a regime-split solver for Black--Scholes implied volatility based on closed-form analytical initial seeds and a universal fourth-order Householder refinement. The method uses a Taylor-reverted at-the-money seed, polynomial Gaussian-CDF seeds in the moderate out-of-the-money region, and a Mills-ratio-corrected seed in the deep out-of-the-money region. The regime boundaries are obtained from CDF truncation tolerances and Householder error-admissibility conditions, rather than from empirical tuning. On two benchmark grids, including extreme-delta cases, the method recovers total volatility to near double-precision accuracy, with maximum absolute errors of order $(10^{-14})$ to $(10^{-13})$. The mean number of Householder updates remains below two, and more than $95\%$ of cases converge within two updates on the granular grid. Native C benchmarks against a state-of-the-art implied-volatility implementation show speedups of about $(1.7)-(1.9\times)$ across several CPU/compiler environments. A Python/Numba implementation gives consistent accuracy and confirms that the regime-split structure is portable beyond hand-optimized C code.
\end{abstract}

\bigskip

\newpage
\section{Introduction}
\label{sec:intro}

The Black--Scholes implied volatility $\sigma$ is the unique solution to
\begin{equation}
c_{\mathrm{BS}}(k,\sigma)=c,\qquad
c_{\mathrm{BS}}(k,v)=
\Phi\left(-\frac{k}{v}+\frac{v}{2}\right)
-e^k\Phi\left(-\frac{k}{v}-\frac{v}{2}\right),
\label{eq:bs}
\end{equation}
where (c) is the normalized undiscounted call price,
$k=\log(K/F)$ is log-moneyness, $v=\sigma\sqrt{T}$ is total volatility,
and $\Phi$ is the standard normal cumulative distribution function. Although
the Black--Scholes map is explicit from volatility to price, the inverse map
from price to implied volatility is not available in elementary closed form and
has motivated a large literature.

Implied-volatility inversion is a basic numerical primitive in option pricing,
volatility-surface construction, calibration and risk calculations. In these
applications the inverse Black--Scholes map may be evaluated millions of times,
so small improvements in per-call speed can produce material gains at the
system level. At the same time, production implementations require robust
accuracy across at-the-money, moderately out-of-the-money and deep-tail regions.
This creates a natural trade-off: explicit approximations are fast but
regime-dependent, while purely numerical solvers are robust but rely heavily on
the quality of the initial value.

The existing literature contains several approaches to this problem, including
explicit approximation formulas, extreme-strike and small-time asymptotics,
uniform bounds, and high-accuracy numerical inversion methods. These
contributions show that different regions of the price--moneyness domain have
different analytical structures. This observation motivates a regime-dependent
initialization strategy rather than a single global seed. The present paper
follows this direction, but with a strict error-control and computational-budget
perspective: the initial seeds are chosen to be algebraic or nearly algebraic,
cheap to evaluate, and accurate enough to make a universal high-order polishing
step effective across the full benchmark domain.

We propose a regime-split seed-plus-solver method. The guiding principle is that
no single elementary seed is uniformly optimal across the full
moneyness--volatility domain. We therefore derive separate analytical seeds in
the domains where their asymptotic assumptions are natural: an at-the-money seed
obtained by Taylor inversion of the exact time-value identity, a moderate
out-of-the-money seed based on polynomial approximations to the Gaussian CDF,
and a deep out-of-the-money seed based on a Mills-ratio correction. These seeds
are then polished by a fourth-order Householder update, with Halley and Newton
steps available as lower-order special cases or fallbacks.

\textbf{In summary, the paper makes the following contributions:}
\begin{itemize}
\item Closed-form seeds cover the at-the-money, moderate out-of-the-money and
deep out-of-the-money regimes.
\item Regime cutoffs are derived from analytical error bounds rather than
empirically tuned values.
\item Mean Householder updates stay below two on both benchmark grids.
\item Native C benchmarks show about $(1.7)-(1.9\times)$ speedups relative
to a widely used high-accuracy benchmark implementation.
\end{itemize}

The resulting method combines a compact analytical seed architecture with a
universal Householder ((d=3)) polishing step. It is designed for
production-style implied-volatility inversion, where speed, robustness,
double-precision accuracy and implementation simplicity are all important.

\section{Literature Review}

Existing implied-volatility inversion methods may be grouped into three broad categories. The first category derives explicit or semi-explicit approximations for implied volatility. Classical examples include the at-the-money approximation of \citet{BS88}, the correction formula of \citet{CM96}, the approximation of \citet{Li2005}, the improved expansion approach of \citet{ChambersNawalkha2001}, and the explicit formula of \citet{SR2017}. These formulas are attractive because they are simple and non-iterative, but their accuracy typically depends on the moneyness and volatility regime in which the approximation is applied.

The second category studies bounds and asymptotic structure. The moment formula of \citet{GaoLee2014} relates the extreme-strike slope of total implied variance to the existence of moments of the underlying asset distribution. Related tail-wing and extreme-strike asymptotics are developed by \citet{BenaimFriz2009}, \citet{Gulisashvili2010}, and \citet{Gulisashvili2012}, while \citet{RoperRutkowski2009} study the relationship between the call-price surface and implied volatility close to expiry. Higher-order model-free asymptotic expansions with explicit error estimates are obtained by \citet{GaoLee2014}. Uniform implied-volatility bounds are studied by \citet{Tehranchi2016}. This literature clarifies why different regions of the price--moneyness domain have different analytical structures, and it motivates regime-dependent initializations rather than a single global approximation.

The third category develops numerical and semi-analytical inversion algorithms. Early Newton-type approaches are connected with the implied-variance calculation of \citet{ManasterKoehler1982}. More recent high-accuracy solvers include the model-free delta-family formula of \citet{CuiKirbyNguyenTaylor2021} and the benchmark implementation of \citet{Jackel2015}. The latter is widely used as a reference implementation because it combines robust normalization, accurate initial approximations and high-order numerical refinement. \citet{Gerhold2012} further shows that implied volatility cannot belong to a broad class of simple closed-form functions, which explains why practical methods typically combine analytical initial approximations with rapidly convergent numerical updates.

A recent contribution by \citet{Schadner2026} proposes a different coordinate system for the inversion problem, showing that Black--Scholes implied volatility can be represented through an inverse-Gaussian quantile in variance space. This is an important theoretical and numerical development: it identifies a distributional transform underlying the inverse Black--Scholes map and reports machine-precision recovery with strong timing performance relative to the benchmark of \citet{Jackel2015}. However, the practical implementation still requires the evaluation of a special-function quantile, so the numerical inversion burden is transferred from the Black--Scholes equation to the evaluation of a probability-law quantile.

The present paper follows a different route. Instead of using one global approximation or transforming the problem into a separate quantile inversion, we construct closed-form analytical seeds tailored to the at-the-money, moderate out-of-the-money and deep out-of-the-money regimes. The regime cutoffs are derived from Taylor-CDF truncation tolerances and Householder error-admissibility conditions. The resulting method therefore keeps the initial seed stage algebraic, while using a universal fourth-order Householder refinement to recover double-precision implied volatility across the benchmark domain.

\section{Methodology}
\label{sec:failure}





\subsection{Regime-Split based Analytical Initial Points}
\label{sec:seed}

\subsubsection{ATM High order Taylor inversion}
\label{sec:atm}
\begin{proposition}
At $k=0$, the normalised Black--Scholes time value ($\ctv$) satisfies,
\begin{equation}
  \ctv = 2\Phi(v/2)-1 = \mathrm{erf}\!\left(\frac{v}{2\sqrt{2}}\right).
  \label{eq:atm-id_in}
\end{equation}
where $\ctv=c-\max(1-e^k,0)$. 
With $s=\sqrt{2\pi}\,\ctv$, an inverted expansion gives (proved in detail in appendix section \ref{sec:app}),
\begin{equation}
  v = s\!\left(1+\frac{s^2}{24}+\frac{7s^4}{1920}+\frac{127s^6}{322560}\right)+O(s^8).
  \label{eq:inverse_in}
\end{equation}
The three ATM seeds used in the benchmark are successive truncations:
\begin{align}
  v_0^{(2)} &= s\!\left(1+\tfrac{s^2}{24}\right),\\
  v_0^{(3)} &= s\!\left(1+\tfrac{s^2}{24}+\tfrac{7s^4}{1920}\right),\\
  v_0^{(4)} &= s\!\left(1+\tfrac{s^2}{24}+\tfrac{7s^4}{1920}+\tfrac{127s^6}{322560}\right).\label{eq:T4_in}
\end{align}
\end{proposition}
\begin{remark}
The first order expansion corresponds to formula in  \citet{BS88}  where $v_{0}=s=\sqrt{2 \pi} c_{tv}$.    
\end{remark}

\subsubsection{A near ATM Initial Value Method}
\label{sec:sr}
 \cite{SR2017} (SR) gives the analytical formulas,
\begin{equation}
  h = \frac{\sqrt{2\pi}}{1+e^k}\!\left(c-\frac{1-e^k}{2}\right),\qquad
  \hat\sigma_{\mathrm{SR}} = h+\sqrt{h^2-k^2/\pi}.
  \label{eq:sr}
\end{equation}


\begin{remark}
 For $|k|\ge\kap$ with $\kap=0.01$, we first attempt the SR formula~\eqref{eq:sr},
which requires no $\Phi^{-1}$ call and is valid when $h^2\ge k^2/\pi$.
Although this method is targeted to give an accurate approximation to IV, we found the values quite useful to feed Halley/Householder solver, in the domain it is defined. Therefore, SR method can be used as an alternative to following method in the next subsection to provide an efficient initial seed for IV. The only caveat is; assessing equation \eqref{eq:sr} one can directly see that
$\hat\sigma_{\mathrm{SR}}$ is undefined whenever $h^2<k^2/\pi$.
\end{remark}
\subsubsection{Mild-OTM Analytical Seed}
The method uses a first order approximation to Gaussian cdf via Taylor expansion \citet{Sandoval2019} then investigates possible analytic approximations to volatility initial seeds for the root solver involving a mild-OTM region.  
\begin{proposition}
An initial $v_0$ formula can be derived using a polynomial approximation to Gaussian cdf where we present in Appendix \ref{sec:app} just as in equation \eqref{eq:vp1_e} and likewise \eqref{eq:vp1}. 
\begin{equation}
  v_{P_1}
  =
  \frac{h+\sqrt{h^2-4AB}}{2B}.
  \label{eq:vp1_e}
\end{equation}
Furthermore using a $4^{th}$ order Taylor expansion to function $exp(x)$ provides following enhancement to the initial seed $v_0$,
\begin{equation}
  v_{P_1}
  =
  \frac{2c+\varepsilon+\sqrt{N}}{2\,a_1\,(2+\varepsilon)}.
  \label{eq:vp1_clean_in}
\end{equation}


where
\begin{equation}
  \varepsilon = k+\frac{k^2}{2}+\frac{k^3}{6}+\frac{k^4}{24},
  \qquad
  N = (2c_{seed}+\varepsilon)^2 - 8\,a_1^2\,k\,\varepsilon\,(2+\varepsilon).
  \label{eq:N_def_in}
\end{equation}
 and
 \begin{equation}
  h = c_{seed}+\frac{\varepsilon}{2},
  \qquad
  A = a_1\,\varepsilon\,k,
  \qquad
  B = \frac{a_1(2+\varepsilon)}{2},
  \label{eq:hab_t4_in}
\end{equation}
\end{proposition}
Now we see that, equation \eqref{eq:vp1_clean_in} is algebraically identical to \eqref{eq:vp1}
under the $4^{th}$ order Taylor substitution; the reformulation requires only two
polynomial quantities in $k$, namely $\varepsilon$ and $2+\varepsilon$,
together with the observable $c$.

In most of the numerical experiments we use this initial seed $v_0$ method where $k$ falls into mild-OTM region. 

\paragraph{Fallback when $N\leq 0$.}
When $N\leq 0$ (a low-volatility artifact occurring when the observed price is
near the lower boundary), the discriminant is clipped to zero and the seed
reduces to
\begin{equation}
  v_{P_1}\big|_{N=0}
  =
  \frac{2c+\varepsilon}{2\,a_1\,(2+\varepsilon)},
  \label{eq:vp1_clip_in}
\end{equation}
which corresponds to taking $h/(2B)$ in the original formulation
\eqref{eq:vp1}.
\paragraph{A Cubic form extension \texorpdfstring{$v_{P3}$}{vP3}}
\begin{proposition}[Closed-form P3 seed]\label{prop:vp3}
After a one-step Newton correction of $\vP$ on equation \eqref{eq:F3} can be written analytically,
\begin{equation}\label{eq:vP3}
    \vPt = \vP\!\left(1 + \frac{\mathcal{G}_3(w)}{\Dden_3(w)}\right)
           \bigg|_{w=\vP} 
\end{equation}
where, with $w = \vP$ and setting $\delta = 2+\eps$,
\begin{align}
  \mathcal{G}_3(w) &=
    -4\eps\kk^3
    + 24\kk\!\left[\eps - \kk\!\left(\tfrac{\eps}{2}+1\right)\right]w^2
    + \left[24\!\left(1+\tfrac{\eps}{2}\right) - 3\eps\kk\right]w^4
    - \delta\,w^6
    - \frac{24h\,w^4}{\aone},
    \label{eq:N_in}\\[4pt]
  \Dden_3(w) &=
    12\eps\kk^3
    - 24\kk\!\left[\eps - \kk\!\left(\tfrac{\eps}{2}+1\right)\right]w^2
    + \left[24\!\left(1+\tfrac{\eps}{2}\right) - 3\eps\kk\right]w^4
    - 3\delta\,w^6.
    \label{eq:D_in}
\end{align}
\end{proposition}
\begin{proof}
    In appendix \ref{sec:app} we gave a detailed proof.
\end{proof}
\paragraph{\texorpdfstring{A high-order polynomial extension $v_{P7}$}{A high-order polynomial extension vP7}}
\begin{proposition}[Closed-form $P_7$ seed]
\label{prop:vp7}
Let $w=v_{P1}$, $\kappa=|k|$, $\delta=2+\varepsilon$, and
$E=e^{\kappa}$.  
Then
\begin{equation}
  v_{P7} = v_{P1}\!\left(1 + \frac{\mathcal{G}_7(w)}{\Dden_7(w)}\right)\Bigg|_{w=v_{P1}},
  \label{eq:vP7}
\end{equation}
where
\begin{equation}
  \mathcal{G}_7(w) = -S_7(w) + Q_7(w) - R_7(w),
  \qquad
  \Dden_7(w) =  7S_7(w) + Q_7(w) - R_7(w) + \Delta_7(w),
  \label{eq:G7D7_in}
\end{equation}
\begin{proof}
In appendix \ref{sec:app} we give detailed proof of equation \eqref{eq:vP7} and the terms $\mathcal{G}_7(w)$, $\mathcal{D}_7(w)$, $\mathcal{R}_7(w)$, $\mathcal{S}_7(w)$ and $\Delta_7(w)$.
\end{proof}
\end{proposition}

\begin{remark}[Shared structure and computational efficiency]
\label{rem:shared7}
The polynomial $Q_7(w)$ appears identically in both $\mathcal{G}_7$ and $\Dden_7$
and need be computed only once, exactly as the polynomial $Q(w)$
of Remark~8.1 for the $P_3$ case.  In practice $\mathcal{G}_7$ and $\Dden_7$ share
four of their nine terms, so the ratio $\mathcal{G}_7/\Dden_7$ requires fewer
operations than two independent polynomial evaluations.
\end{remark}

The propositions \ref{prop:vp3}, \ref{prop:vp7} and equations \eqref{eq:vP3}, \eqref{eq:vP7} provide an analytic representation to a robust approximation for IV in a mild-OTM region.  
\subsubsection{Corrected Mills-Ratio (deep-OTM) seed}
\label{sec:ratio}

For the deep OTM options, we apply a one-step
\citet{Mills26}, \citet{AbramowitzStegun64} correction to the quadratic seed.

The inversion is always applied to an OTM-equivalent price.
For $k>0$ (call already OTM) we set $\cseed=c$.
For $k<0$ (ITM call), put--call parity gives
$c(k,v)-(1-e^k)=e^k c(|k|,v)$, so the OTM-equivalent price is
\begin{equation}
  \cseed = e^{-k}(c-1+e^k), \qquad k<0.
  \label{eq:cseed}
\end{equation}
Throughout we write $\kap=|k|$.

The simple quadratic seed solves $\cseed\approx\Phi(d_1)$:
\begin{equation}
  v_q = z + \sqrt{z^2+2\kap}, \qquad z=\Phi^{-1}(\cseed).
  \label{eq:vq}
\end{equation}
The Mills-ratio analysis (Section~\ref{sec:mills}) shows that for an OTM call,
\begin{equation}
  c_{\mathrm{BS}}(k,v)
  \approx \Phi(d_1)\cdot\frac{v^2}{\kap+v^2/2}.
  \label{eq:mills-approx}
\end{equation}
Evaluating the correction factor at $v_q$ gives $\alpha_q=v_q^2/(\kap+v_q^2/2)$.
The corrected effective price is $\cseed/\alpha_q$, and the ratio-corrected seed is
\begin{equation}
  z_q = \Phi^{-1}(\cseed/\alpha_q), \qquad
  v_{q,1} = z_q + \sqrt{z_q^2+2\kap}.
  \label{eq:vq1}
\end{equation}
We call this deep-OTM method as ratio correction method.
\paragraph{\texorpdfstring{A fully analytical form for $v_{q,1}$}{A fully analytical form for v{q1}}}
\label{sec:soranzo_vq1}

The ratio-corrected deep-OTM seed $v_{q,1}$ of equation~\eqref{eq:vq1}
requires two evaluations of $\Phi^{-1}$.  The current implementation
uses a piecewise rational approximation due to Acklam's algorithm \citet{acklam2003}, which
is highly efficient because it avoids any transcendental call in the
central probability region.  We note here that an explicit algebraic
form of $\Phi^{-1}$, and hence of $v_{q,1}$, exists via the
invertible Gaussian CDF approximation of~ \citet{Soranzo2012}.
\begin{proposition}[Analytical $v_{q,1}$]
\label{prop:vq1_SE}
For $\kappa=|k|$ and $c_{\mathrm{seed}}<\tfrac{1}{2}$, define
\begin{align}
  L_1   &= \log\!\bigl(4c_{\mathrm{seed}}(1-c_{\mathrm{seed}})\bigr), \nonumber\\ 
  u_1   &= u^+(L_1),\nonumber
  \\ 
  z     &= -\sqrt{u_1},\\ 
  v_q   &= z + \sqrt{u_1 + 2\kappa},
  \\ \nonumber
  \alpha_q &= \frac{v_q^2}{\kappa + v_q^2/2},
  \\ \nonumber
  L_2   &= \log\!\Bigl(4\,\tfrac{c_{\mathrm{seed}}}{\alpha_q}
           \Bigl(1-\tfrac{c_{\mathrm{seed}}}{\alpha_q}\Bigr)\Bigr),
  \\ \nonumber
  u_2   &= u^+(L_2),
  \\
  v_{q,1} &= -\sqrt{u_2} + \sqrt{u_2 + 2\kappa}.
  \label{eq:vq1_SE_final}
\end{align}
\end{proposition}
where $u^{+}(L)$ is defined in proposition \ref{prop:vq1_SE} at the appendix \ref{sec:app}.

 \subsection{Full Initial IV point definition}
 \label{sec:fullseed}
\begin{definition}
\label{def:regime_seed}
Let $c_{\mathrm{tv}} = c - \max(1-e^k,0)$ and $v_{\mathrm{ATM}} = v_0^{(4)}$.
For $|k|\ge\kappa$ define the OTM-equivalent price
\[
  c_{\mathrm{seed}} =
  \begin{cases}
    e^{-k}\!\left(c - 1 + e^k\right) & k < 0,\\
    c                                  & k \ge 0,
  \end{cases}
\]
Define
\begin{equation}
  v_0 =
  \begin{cases}
    v_{\mathrm{ATM}}(c)
      & |k| < \kappa_{1},\\[6pt]
    \max\!\bigl(v_{P7}(|k|,\,c_{\mathrm{seed}}),\;v_{\mathrm{ATM}}\bigr)
      & \kappa_{1} \le |k| \le \kappa_{2},\\[6pt]
    \max\!\Bigl(\tfrac{1}{2}\bigl(v_{P3}+v_{P7}\bigr)(|k|,\,c_{\mathrm{seed}}),\;v_{\mathrm{ATM}}\Bigr)
      & \kappa_{2} < |k| \le \kappa_{3},\\[6pt]
    \max\!\bigl(v_{P3}(|k|,\,c_{\mathrm{seed}}),\;v_{\mathrm{ATM}}\bigr)
      & \kappa_{3}< |k| \le \kappa_{4},\\[6pt]
    \max\!\bigl(v_{q,1}(|k|,\,c_{\mathrm{seed}}),\;v_{\mathrm{ATM}}\bigr)
      & |k| > \kappa_{4}.\\
      v_{q,1}(|k|,\,c_{\mathrm{seed}})&  c \leq c^{*} ~~\&~ |k| \geq \kappa^{*} 
  \end{cases}
  \label{eq:regime_seed}
\end{equation}
where $\kappa^{*}=0.5,c^{*}=0.02128,\kappa_1=0.001,\kappa_2=0.81,\kappa_3=1.155, \kappa_4=1.347$.
\end{definition}

\begin{remark}
    All the regime seed boundaries $\kappa_1, \kappa_2, \kappa_3, \kappa_4$ and $c^{*},\kappa^{*}$ will be computed based on theoretical implied errors and  tolerance chosen for  Householder $(d=3)$ solver, which we will derive later sections. 
\end{remark}
\subsection{Solver Methods}
\paragraph{Halley's method}

We finally wrap-up IV computation methodology by presenting the inversion algorithm. As noted in \eqref{eq:bs}, c is normalised call price With $d_1=-k/v+v/2$, $d_2=d_1-v$, $\phi=\Phi'$. Then the Halley step is, 
\begin{equation}
  v \leftarrow v - \frac{2fV}{2V^2-fW},\qquad
  f=c_{\mathrm{BS}}-c,\quad V=\phi(d_1),\quad W=\phi(d_1)d_1d_2/v.
  \label{eq:halley}
\end{equation}
This method has cubic convergence therefore a decent initial value seed having error around $10^{-1}$ could provide convergence at the level of $10^{-14}$ after 3-4 iterations. 
\begin{proposition}[Cubic convergence]
Suppose $f\in C^3$ on a neighbourhood $I$ of $v^*$ with $f'(v^*)\neq 0$.
Then the Halley iterates satisfy
\begin{equation}
e_{n+1} = -\frac{1}{6}\!\left(\frac{f'''(v^*)}{f'(v^*)}
- \frac{3}{2}\!\left[\frac{f''(v^*)}{f'(v^*)}\right]^{\!2}\right)e_n^3
+ O(e_n^4).
\end{equation}
In particular $|e_{n+1}| \leq C|e_n|^3$ for some $C>0$, so the method is
cubically convergent.
\end{proposition}
\paragraph{\texorpdfstring{Householder's Method (HH-4) $d=3$}{Householder's Method  d=3}}

Halley's method exploits $f''$ to achieve cubic convergence; including $f'''$
yields a fourth-order method \citet{Traub82}.
With the same notation as \eqref{eq:halley}, define the scaled derivatives

\begin{equation}
  \alpha = \frac{f''(v)}{f'(v)} = \frac{d_1 d_2}{v},
  \qquad
  \beta  = \frac{f'''(v)}{f'(v)} = \frac{(d_1 d_2)^2 - (d_1^2+d_2^2) - d_1 d_2}{v^2},
  \label{eq:hh4-derivs}
\end{equation}

where both identities follow from the exact relations
$\partial d_1/\partial v = -d_2/v$ and $\partial d_2/\partial v = -d_1/v$,
together with $f' = \phi(d_1)$.
Setting $r = f/f' = f/V$, the Householder step is

\begin{equation}
  v \;\leftarrow\; v + \frac{3r\!\left(2 - r\alpha\right)}
                             {-6 + 6r\alpha - r^2\beta},
  \label{eq:hh4}
\end{equation}

with a Newton fallback $v \leftarrow v - r$ when the cubic denominator
is smaller than $10^{-20}$ in absolute value.
The step \eqref{eq:hh4} reduces to the Halley step \eqref{eq:halley}
when $\beta=0$, and to the Newton step when both $\alpha=\beta=0$.

\begin{proposition}[Quartic convergence of Householder $d=3$]{\label{prop:quartichh4}}
Let $f \in C^4$ on a neighbourhood $I$ of a simple root $v^*$,
i.e.\ $f(v^*)=0$ and $f'(v^*)\neq 0$.
Then the iterates produced by \eqref{eq:hh4} satisfy
\begin{equation}
  e_{n+1} = -\frac{1}{24}
  \left(
    \frac{f^{(4)}(v^*)}{f'(v^*)}
    - 4\,\frac{f'''(v^*)}{f'(v^*)}\cdot\frac{f''(v^*)}{f'(v^*)}
    + 3\!\left[\frac{f''(v^*)}{f'(v^*)}\right]^{\!3}
  \right)e_n^4
  + O(e_n^5).
  \label{eq:hh4-conv}
\end{equation}
In particular $|e_{n+1}|\leq C|e_n|^4$ for some $C>0$,
so the method is quartically convergent.
\end{proposition}



The details can be found in standard references
e.g.~\citet{Traub82}, \citet{Stoer02}.
In the following section~\ref{sec:IValgo} we present the complete
pseudo-algorithm used in computations. For the rest of the study we will shortly call this method as HH-4 method.

\subsection{IV Solver Pseudo Algorithm}\label{sec:IValgo}
\begin{algorithm}[ht]
\footnotesize\DontPrintSemicolon
\caption{Regime-Split Implied Volatility Solver}
\label{alg:regime_solver}
\KwIn{normalised call price $c$, log-moneyness $k$, tolerance $\epsilon=10^{-14}$}
\KwOut{$v^*$ such that $|c_{\mathrm{BS}}(k,v^*)-c|<\epsilon$}

$s\leftarrow\sqrt{2\pi}[c{-}\max(1{-}e^k,0)]$;\;
$v_{\mathrm{ATM}}\leftarrow s(1+s^2/24+7s^4/1920+127s^6/322560)$\tcp*[r]{eq.~\eqref{eq:inverse}}
$\kappa\leftarrow|k|$;\;
$c_{\mathrm{seed}}\leftarrow e^{-k}(c{-}1{+}e^k)$ if $k<0$, else $c$\tcp*[r]{eq.~\eqref{eq:cseed}}
\BlankLine
\uIf{$\kappa < \kappa_1$}{
  $v_0\leftarrow v_{\mathrm{ATM}}$\tcp*[r]{ATM, $\kappa_1=0.001$}
}
\uElseIf{$c_{\mathrm{seed}}<c^*$ \textbf{and} $\kappa>\kappa^*$}{
  $v_0\leftarrow v_{q,1}(\kappa,\,c_{\mathrm{seed}})$\tcp*[r]{tail guard, eq.~\eqref{eq:vq1}, $c^*{=}0.02128$, $\kappa^*{=}0.5$}
}
\uElseIf{$\kappa\le\kappa_4$}{
  $\varepsilon\leftarrow k{+}k^2/2{+}k^3/6{+}k^4/24$\;
  $v_{P1}\leftarrow(2c_s{+}\varepsilon{+}\sqrt{\max(N,0)})/(2a_1(2{+}\varepsilon))$\tcp*[r]{eq.~\eqref{eq:vp1}}
  $E\leftarrow e^\kappa$\;
  $v_{P7}\leftarrow v_{P1}\bigl(1+G_7(v_{P1})/D_7(v_{P1})\bigr)$\tcp*[r]{eq.~\eqref{eq:vP7}}
  \uIf{$\kappa\le\kappa_1$}{
    $v_0\leftarrow\max(v_{P7},v_{\mathrm{ATM}})$\tcp*[r]{$\kappa_2{=}0.81$, $\varepsilon_\Phi{=}10^{-3}$}
  }\Else{
    $v_{P3}\leftarrow v_{P1}\bigl(1+G_3(v_{P1})/D_3(v_{P1})\bigr)$\tcp*[r]{eq.~\eqref{eq:vP3}}
    \uIf{$\kappa\le\kappa_3$}{
      $v_0\leftarrow\max(\tfrac12(v_{P3}{+}v_{P7}),v_{\mathrm{ATM}})$\tcp*[r]{avg, $\kappa_3{=}1.155$, $\varepsilon_\Phi{=}5{\times}10^{-3}$}
    }\Else{
      $v_0\leftarrow\max(v_{P3},v_{\mathrm{ATM}})$\tcp*[r]{$\kappa_4{=}1.347$, $\varepsilon_\Phi{=}10^{-2}$}
    }
  }
}
\Else{
  $v_0\leftarrow v_{q,1}(\kappa,\,c_{\mathrm{seed}})$\tcp*[r]{deep-OTM, eq.~\eqref{eq:vq1}}
}
\BlankLine
$v\leftarrow v_0$;\quad $e^k$ precomputed once\;
\Repeat(\tcp*[f]{HH-4 polish, eq.~\eqref{eq:hh4}, all regimes}){$|f|<\epsilon$}{
  $d_1\leftarrow{-}k/v{+}v/2$;\;
  $d_2\leftarrow d_1{-}v$;\;
  $f\leftarrow\Phi(d_1){-}e^k\Phi(d_2){-}c$;\;
  $r\leftarrow f/\Phi'(d_1)$;\;
  $\alpha\leftarrow d_1d_2/v$;\;
  $\beta\leftarrow[(d_1d_2)^2{-}(d_1^2{+}d_2^2){-}d_1d_2]/v^2$\;
  $v\leftarrow v+3r(2{-}r\alpha)/({-}6{+}6r\alpha{-}r^2\beta)$\tcp*[r]{eq.~\eqref{eq:hh4}}
}
\Return{$v$}
\end{algorithm}
\begin{remark}
    The joint price and moneyness filter in line 7 corresponds to $c^{*}=0.02128$ and $\kappa^{*}=0.5$ derived in appendix \ref{sec:priceandotmtail}. Then, essential regime barriers, $\kappa_{1} \dots \kappa_{4}$ are defined under section \ref{sec:regimeB}. 
\end{remark}
\FloatBarrier
\subsection{Theoretical justification of Regime Boundaries }
\label{sec:asymptotics}

We establish that each component of regime boundaries used in \ref{def:regime_seed} has a theoretical underpinning through truncation of Gaussian CDF and chosen error tolerance.

\subsubsection{ Taylor expansion of Gaussian distribution CDF and truncation tolerance to regime boundaries}
\label{sec:boundaries}

The Taylor--CDF approximation $P_j$ is trusted while the argument
$x$ satisfies $|a_{2m+1}x^{2m+1}|\le\varepsilon_\Phi$, i.e.\
$|x|\le a_\eta$ where
\begin{equation}
  a_\eta(\varepsilon_\Phi)
  = \left(\frac{\varepsilon_\Phi}{|a_{2m+1}|}\right)^{1/(2m+1)}.
  \label{eq:a_eta}
\end{equation}
For $P_7$ the first neglected term is $a_9 x^9$, so
\[
  a_{\eta,P7}(\varepsilon_\Phi)
  = \left(\frac{\varepsilon_\Phi}{|a_9|}\right)^{1/9}.
\]
The OTM condition $|d_2|=\kappa/v+v/2\le a_\eta$ gives the
log-moneyness validity boundary
\begin{equation}
  k_{\max}(v;\,a_\eta) = a_\eta v - \tfrac{v^2}{2}.
  \label{eq:kmax}
\end{equation}

\begin{proposition}[Concavity of the Taylor--CDF boundary]
\label{prop:concavity}
For fixed $a_\eta>0$, the function $k_{\max}(v;a_\eta)=a_\eta v-v^2/2$
is strictly concave in $v$, attains its maximum at $v^*=a_\eta$, and
\[
  \max_{v>0}\,k_{\max}(v;\,a_\eta) = \frac{a_\eta^2}{2}.
\]
\end{proposition}

\begin{proof}
Differentiating twice: $k_{\max}'(v)=a_\eta-v$, $k_{\max}''(v)=-1<0$.
Hence the function is strictly concave, the unique maximum is at
$v^*=a_\eta$, and $k_{\max}(a_\eta;\,a_\eta)=a_\eta^2/2$.
\end{proof}%
The peak log-moneyness implied by CDF tolerance $\varepsilon_\Phi$ for
$P_7$ is therefore
\begin{equation}
  k_{\max}^{\mathrm{peak}}(\varepsilon_\Phi)
  = \frac{a_{\eta,P7}(\varepsilon_\Phi)^2}{2}
  = \frac{1}{2}\!\left(\frac{\varepsilon_\Phi}{|a_9|}\right)^{2/9}.
  \label{eq:kpeak}
\end{equation}

\subsubsection{HH4 admissibility and Transition initial seed selection }
\label{sec:transition}
\paragraph{HH-4 admissibility condition}

The admissibility condition for two net HH-4 corrections is
derived as follows.  Let $\eta_v=|v_0/v-1|$ be the relative seed
error.  If the HH-4 update is applied for two correction steps to
tolerance $\varepsilon_v$, the idealised admissible seed level is
\begin{equation}
  \eta_v^* = \varepsilon_v^{1/16}.
  \label{eq:hh4_admissible}
\end{equation}
For $\varepsilon_v=10^{-14}$, equation~\eqref{eq:hh4_admissible} gives
$\eta_v^*\approx13.3\%$. 

This follows directly from two applications of the quartic contraction bound in 
Proposition~\ref{prop:quartichh4}: $|e_1| \leq C|e_0|^4$ and 
$|e_2| \leq C|e_1|^4 \leq C^5|e_0|^{16}$. Requiring $|e_2| \lesssim \varepsilon_v$ 
and dropping the order-unity constant $C^5$ yields $\eta_v^* = \varepsilon_v^{1/16}$.
This rule marginally surpasses the two-iteration performance standard established
by~\citet{Jackel2015} as the hallmark of a high-quality IV
initialisation.
\paragraph{Initial Seed Selection}{\label{sec:regimeB}}
Substituting $|a_9|=1/(3456\sqrt{2\pi})$ in equation \eqref{eq:kpeak} yields the three regime boundaries, namely, $k_{\max}^{\mathrm{peak}}$:
\begin{equation}
  \varepsilon_\Phi = 10^{-3}
    \;\Longrightarrow\;
    \kappa_{2} \approx 0.81,
  \qquad
  \varepsilon_\Phi = 5 \times 10^{-3}
    \;\Longrightarrow\;
    \kappa_{3} \approx 1.155, \qquad 
  \varepsilon_\Phi = 10^{-2}
    \;\Longrightarrow\;
    \kappa_{4} \approx 1.347.
  \label{eq:boundaries}
\end{equation}
The initial moneyness $\kappa_1=10^{-3}$ is for near ATM boundary, which is arbitrary and might be fine tuned for different grids.

The interval $0.81<|k|\le1.347$ is therefore the transition region
between the stricter ($10^{-3}$) and looser ($10^{-2}$) $P_7$
validity domains, and the cutoffs are analytical consequences of
\eqref{eq:kpeak}, not empirically tuned constants.

Inside the transition interval we define the weighted average seed
\begin{equation}
  v_{\mathrm{avg}} = \omega v_{P3}+ \left( 1-\omega \right) v_{P7}.
  \label{eq:vavg}
\end{equation}
We then immediately select $\omega=\frac{1}{2}$.
The averaged transition seed \eqref{eq:vavg} is then directly in
closed form:
\begin{equation}
  v_{\mathrm{avg}}
  = \frac{v_{P3}+v_{P7}}{2}
  = v_{P1}\!\left(
      1 + \frac{1}{2}\!\left(\frac{\mathcal{G}_3(w)}{\Dden_3(w)}
                            +\frac{\mathcal{G}_7(w)}{\Dden_7(w)}\right)
    \right)\Bigg|_{w=v_{P1}},
  \label{eq:vavg_closed}
\end{equation}
where $\mathcal{G}_3$, $\Dden_3$ are as in Proposition~2.5 and $\mathcal{G}_7$, $\Dden_7$ as in
Proposition~\ref{prop:vp7}.  Equation~\eqref{eq:vavg_closed} expresses
$v_{\mathrm{avg}}$ as a single rational function of $(k, c, \varepsilon)$
evaluated at $w=v_{P1}$, with no iterative component.
The weight $\omega=\tfrac12$ is justified by the following two-regime
analysis. 

The $\eta_{0}^{\max}=\%8.67$ is implied relative error from $v_{avg}$. Since $\eta_0^{\max}=8.67\%<13.3\%=\eta_v^*$,
the averaged transition seed \eqref{eq:vavg} satisfies the two-HH4
admissibility condition throughout the entire transition interval.
\begin{remark}[Sign structure of the transition interval]
\label{rem:sign}
Numerical evaluation of $v_{P3}/v$ and $v_{P7}/v$ on the transition
grid (Table~\ref{tab:ratio}) shows two sub-regimes:
\begin{itemize}
\item \emph{Small price and mild-to-deep OTM tail filter}
As derived the datils in appendix \ref{sec:priceandotmtail}, we apply a sub-filter which corrects deficiencies of current multi-regime moneyness filters with also a small price rule based on HH4 minimum two iteration admissibility condition. This corresponds to a joint condition of $c_{seed}<c^{*}$ and $k>k^{*}$.
  \item \emph{Bracketing sub-regime} ($\kappa_{2}<|k|\lesssim \kappa_{3}$):
        $v_{P3}/v>1$ and $v_{P7}/v<1$, so $v_{P3}$ overshoots and
        $v_{P7}$ undershoots the true root.  The midpoint
        $v_{\mathrm{avg}}$ is then a true bracket of the root, and
        $\omega=\tfrac12$ is the unbiased, parameter-free choice within
        the bracket. The upper boundary in this regime choice corresponds to \%0.5 CDF error tolerance for Taylor expansion of Gaussian CDF and lower boundary corresponds to \%0.1.
  \item \emph{Error-reduction sub-regime} ($\kappa_{3}\lesssim|k|\le \kappa_{4}$):
        both seeds undershoot ($v_{P3}/v<1$, $v_{P7}/v<1$), but $v_{P3}$
        decays more slowly.  Therefore, in this regime there's a revert-back to $vp_3$ initial seed in this region. The upper boundary choice corresponds to \%1 CDF error tolerance for Taylor expansion of Gaussian CDF.
\end{itemize}
In both sub-regimes the maximum relative seed error of $v_{\mathrm{avg}}$
is $\eta_0^{\max}=8.67\%$ (Table~\ref{tab:ratio}).
\end{remark}

\subsubsection{Summary of the analytical boundary construction and tolerance decisions}
\label{sec:chain_summary}

The complete derivation may be summarised as five linked steps, each
flowing from the previous with no free parameters:

\begin{enumerate}
  \item \textbf{CDF truncation.}
    Truncate $\Phi$ at orders 1, 3, 7 to obtain $P_1$, $P_3$, $P_7$
    with known first neglected terms $a_9 x^9$ (for $P_7$).

  \item \textbf{Seed derivation.}
    Substitute $P_1$ into the BS equation and solve the resulting
    quadratic analytically ($v_{P1}$, equation~\eqref{eq:vp1}).
    One Newton step on the $P_3$/$P_7$ surrogates yields $v_{P3}$,
    $v_{P7}$ in closed form (equation~\eqref{eq:vP7}).

  \item \textbf{Tolerance-to-boundary.}
    Impose $|a_9 x^9|\le\varepsilon_\Phi$ on the argument
    $x=d_2$ to obtain the admissible argument size $a_\eta$
    (equation~\eqref{eq:a_eta}) and the concave boundary
    $k_{\max}(v)=a_\eta v-v^2/2$ (Proposition~\ref{prop:concavity}).
    Two round-number tolerances $\varepsilon_\Phi\in\{10^{-3},10^{-2}\}$
    give peak boundaries $\kappa_{2}$ and $\kappa_{4}$ analytically
    (equation~\eqref{eq:boundaries}).

  \item \textbf{Admissibility verification.}
    The averaged seed $v_{\mathrm{avg}}=\tfrac12(v_{P3}+v_{P7})$
    achieves maximum relative seed error $\eta_0^{\max}=8.67\%$
    (Table~\ref{tab:ratio}), which lies below the two-HH4
    admissibility threshold $\eta_v^*\approx13.3\%$
    (equation~\eqref{eq:hh4_admissible}).  Mean net HH4 updates
    are below $2$ on both benchmark grids (Table~\ref{tab:iters}).

  \item \textbf{Price-residual reconciliation.}
    The Vega--Lipschitz bound (Lemma~\ref{lem:lipschitz}) connects
    the seed-error view to the code-level stopping criterion:
    $\eta_v\le\eta_v^*$ implies
    $|f(v_0)|\lesssim\eta_v^*v^*\varphi(d_1)$, after which quartic
    HH4 convergence reaches $|f|<10^{-14}$ in at most two net steps. This method is used for the determination of near OTM price filter to locate some of the edge cases to more appropriate regime.
\end{enumerate}

\noindent
Steps~1--3 establish the regime boundaries from first principles;
steps~4--5 certify that the chosen seeds satisfy the convergence
requirements of the polisher.  Together they provide a unified
theoretical justification for the regime-split architecture used in the study.

\newpage
\section{Results}
\label{sec:bench}

\subsection{Experimental setup}

\paragraph{Grids.}
\begin{itemize}
    \item In first grid, we use 328 base options, with volatilities
    $v\in\{0.01,0.05,0.10,\ldots,2.00\}$ and call deltas
    $\Delta\in\{0.05,0.20,0.30,0.45,0.55,
    0.70,0.80,0.95\}$, tiled 5,000 times, producing
    1,640,000 IV inversions.

    \item In second, we use 1,970 base options, with volatilities
    $v\in\{0.01,0.05,0.06,\ldots,2.00\}$ and call deltas
    $\Delta\in\{0.01,0.05,0.20,0.30,0.45,0.55,
    0.70,0.80,0.95,0.99\}$, tiled 5,000 times, producing
    9,850,000 IV inversions.
\end{itemize}
\paragraph{Accuracy target.}
Maximum absolute error against the known grid volatility;
all timings in nanoseconds per implied volatility (ns/IV).
\subsection{ATM seed convergence before root solver update step}

Table~\ref{tab:atm_seed_values} reports unpolished seed values;
Table~\ref{tab:atm_seed_abs_errors} gives relative errors.
Taylor4 achieves relative error $1.19\times10^{-10}$ at $v=0.20$. 

\begin{table}[!htbp]
\centering
\caption{ATM seed values $v_0$ before HH-4 method.}
\label{tab:atm_seed_values}
\begin{tabular}{lcccc}
\toprule
$v$ & Bren--S $s$ & Taylor2 & Taylor3 & Taylor4 \\
\midrule
0.01 & 0.0099999583 & 0.0100000000 & 0.0100000000 & 0.0100000000 \\
0.05 & 0.0499947922 & 0.0499999989 & 0.0500000000 & 0.0500000000 \\
0.20 & 0.1996671661 & 0.1999988380 & 0.1999999950 & 0.2000000000 \\
0.40 & 0.3973492574 & 0.3999632598 & 0.3999993724 & 0.3999999882 \\
0.90 & 0.8705258291 & 0.8980132342 & 0.8998358919 & 0.8999850562 \\
1.20 & 1.1317270359 & 1.1921238223 & 1.1988925223 & 1.1998287577 \\
1.50 & 1.3704872161 & 1.4777412721 & 1.4953680022 & 1.4989433486 \\
2.00 & 1.7112487838 & 1.9200476883 & 1.9735487713 & 1.9904681838 \\
\bottomrule
\end{tabular}
\end{table}
\begin{table}[!htbp]
\centering
\caption{ATM seed relative errors $|v_0-v|/v$ before Halley/Householder $4^{th}$ order method.}
\label{tab:atm_seed_abs_errors}
\small
\begin{tabular}{lcccc}
\toprule
$v$ & B--S $s$ & Taylor2 & Taylor3 & Taylor4 \\
\midrule
0.01 & $4.167\times 10^{-8}$  & $3.647\times 10^{-13}$ & $1.665\times 10^{-16}$ & $1.613\times 10^{-16}$ \\
0.05 & $5.208\times 10^{-6}$  & $1.139\times 10^{-9}$  & $3.075\times 10^{-13}$ & $1.388\times 10^{-16}$ \\
0.20 & $3.328\times 10^{-4}$  & $1.162\times 10^{-6}$  & $5.005\times 10^{-9}$  & $2.384\times 10^{-11}$ \\
0.40 & $2.651\times 10^{-3}$  & $3.674\times 10^{-5}$  & $6.276\times 10^{-7}$  & $1.185\times 10^{-8}$ \\
0.90 & $2.947\times 10^{-2}$  & $1.987\times 10^{-3}$  & $1.641\times 10^{-4}$  & $1.494\times 10^{-5}$ \\
1.20 & $6.827\times 10^{-2}$  & $7.876\times 10^{-3}$  & $1.107\times 10^{-3}$  & $1.712\times 10^{-4}$ \\
1.50 & $1.295\times 10^{-1}$  & $2.226\times 10^{-2}$  & $4.632\times 10^{-3}$  & $1.057\times 10^{-3}$ \\
2.00 & $2.888\times 10^{-1}$  & $7.995\times 10^{-2}$  & $2.645\times 10^{-2}$  & $9.532\times 10^{-3}$ \\
\bottomrule
\end{tabular}
\end{table}
\newpage
\subsection{\texorpdfstring{Mild-OTM Region with $\eta$ method}{Mild-OTM Region with eta method}}
Following table \ref{tab:p7-p3-average} summarizes the performance of our polynomial methods for initial seed computation for mild-OTM region.
\begin{table}[ht]
\centering
\caption{Transition-region seed quality: pure $v_{P7}$ versus averaged seed
$v_{\rm avg}=\frac12(v_{P3}+v_{P7})$. The average is applied for
$0.81<|k|< 1.347$ and $1.28$ is the last point inside HH-4 admissibility criteria.}
\label{tab:p7-p3-average}
\small
\begin{tabular}{rrrrrrrr}
\toprule
$k$ & $v$ & $v_{P7}/v$ & $\eta_{P7}$ & HH4
& $v_{\rm avg}/v$ & $\eta_{\rm avg}$ & HH4 \\
\midrule
0.550 & 0.400 & 1.0214 & 2.14\% & 2 & 1.0214 & 2.14\% & 2 \\
0.599 & 0.449 & 1.0050 & 0.50\% & 2 & 1.0050 & 0.50\% & 2 \\
0.647 & 0.497 & 0.9916 & 0.84\% & 2 & 0.9916 & 0.84\% & 2 \\
0.696 & 0.546 & 0.9801 & 1.99\% & 2 & 0.9801 & 1.99\% & 2 \\
0.745 & 0.595 & 0.9697 & 3.03\% & 2 & 0.9697 & 3.03\% & 2 \\
0.793 & 0.643 & 0.9599 & 4.01\% & 2 & 0.9599 & 4.01\% & 2 \\
\midrule
0.842 & 0.692 & 0.9503 & 4.97\% & 2 & 1.0068 & \textbf{0.68\%} & 2 \\
0.891 & 0.741 & 0.9408 & 5.92\% & 2 & 0.9963 & \textbf{0.37\%} & 2 \\
0.939 & 0.789 & 0.9312 & 6.88\% & 2 & 0.9861 & \textbf{1.39\%} & 2 \\
0.988 & 0.838 & 0.9213 & 7.87\% & 2 & 0.9760 & \textbf{2.40\%} & 2 \\
1.037 & 0.887 & 0.9112 & 8.88\% & 2 & 0.9658 & \textbf{3.42\%} & 2 \\
1.085 & 0.935 & 0.9007 & 9.93\% & 2 & 0.9556 & \textbf{4.44\%} & 2 \\
\midrule
1.134 & 0.984 & 0.8900 & 11.00\% & 2 & 0.9453 & \textbf{5.47\%} & 2 \\
1.183 & 1.033 & 0.8789 & 12.11\% & 2 & 0.9348 & \textbf{6.52\%} & 2 \\
1.231 & 1.081 & 0.8676 & 13.24\% & 2 & 0.9242 & \textbf{7.58\%} & 2 \\
\itshape 1.280 & \itshape 1.130 
      & \itshape 0.9706 & \itshape \textbf{2.94 \%} & 2 
      & \itshape 0.9133 & \itshape \textbf{8.67\%} & 2\\
\bottomrule
\end{tabular}
\end{table}
\FloatBarrier

\subsection{Deep OTM Region approximation with Ratio Correction Method }

Table~\ref{tab:ratio} shows a representative sample where Mills-ratio-correction method is best performing. In this region ratio-correction regime is applied as per the algorithm outlined in \ref{sec:IValgo}. We see that updated root $v_{q,1}$ provides a dramatic precision update to the initial seed and we empirically observe that the asymptotic theory works quite well even after $k=1.2$ in this setting.
\begin{table}[h]
\centering
\caption{Initial seed accuracy on the deep-OTM sample. 
$v_q$ denotes the plain quadratic seed and $v_{q,1}$ denotes the 
ratio-corrected seed. The ratio $v_{q,1}/v$ shows how close the corrected 
seed is to the exact volatility. The last column reports the percentage 
error reduction relative to $v_q$.}
\label{tab:ratio}
\small
\begin{tabular}{@{}rrrrrrr@{}}
\toprule
$k$ & $v$ & $v_q$ & $v_{q,1}$ & $v_{q,1}/v$ 
& $|v_q-v|$ & $|v_{q,1}-v|$ (improv.) \\
\midrule
1.347 & 0.650 & 0.5102 & 0.6744 & 1.0375 & 0.1398 & 0.0244 \ (83\%) \\
1.619 & 1.350 & 1.0323 & 1.3406 & 0.9930 & 0.3177 & 0.0094 (97\%) \\
2.047 & 1.350 & 1.0866 & 1.3521 & 1.0016 & 0.2634 & 0.0021 (99\%) \\
3.592 & 1.500 & 1.3216 & 1.5032 & 1.0021 & 0.1784 & 0.0032 (98\%) \\
4.075 & 1.650 & 1.4717 & 1.6502 & 1.0001 & 0.1783 & 0.0003 (99\%) \\
5.290 & 2.000 & 1.8246 & 1.9948 & 0.9974 & 0.1754 & 0.0052 (97\%) \\
\bottomrule
\end{tabular}
\end{table}

\begin{table}[ht]
\centering
\caption{Deep-OTM seed accuracy: exact $\Phi^{-1}$ versus
Soranzo--Epure algebraic $\Phi^{-1}$.  Both seeds satisfy the
HH4 two-correction admissibility condition $\eta_v<13.3\%$.}
\label{tab:SE_accuracy}
\small
\begin{tabular}{@{}rrrrrr@{}}
\toprule
$k$ & $v$ & $c_{\mathrm{seed}}$ & $v_{q,1}$ & $|v_{q,1}^{\mathrm{SE}}-v_{q,1}|$ & $\eta_v^{\mathrm{SE}}$ (\%) \\
\midrule
1.347 & 0.650 & 0.008533 & 0.674355 & $6.6\times10^{-6}$ & 3.748 \\
1.619 & 1.350 & 0.146350 & 1.340638 & $3.9\times10^{-5}$ & 0.691 \\
2.047 & 1.350 & 0.090000 & 1.352130 & $4.3\times10^{-6}$ & 0.158 \\
3.592 & 1.500 & 0.019835 & 1.503179 & $5.1\times10^{-6}$ & 0.212 \\
4.075 & 1.650 & 0.021024 & 1.650248 & $1.9\times10^{-6}$ & 0.015 \\
5.290 & 2.000 & 0.023467 & 1.994842 & $5.1\times10^{-6}$ & 0.258 \\
\bottomrule
\end{tabular}
\end{table}
Table~\ref{tab:SE_accuracy} shows the seed error $\eta_{v} = |v_{q,1}/v-1|$
on the deep-OTM benchmark points.  The Soranzo--Epure seed matches the
rational $\Phi^{-1}$ seed to within $7\times10^{-6}$ across the full
deep-OTM range, confirming that the approximation error is negligible
relative to the HH-4 admissibility threshold of $13.3\%$.
\FloatBarrier
\subsection{Regime based seed overall performance}
 In this table \ref{tab:seed_quality} we present an overall performance and veiew of the initial seed behaviour across all regimes. We emphasize some specific values which are near edge within each regime. An interesting observation is; We immediately see that, when the $\eta$ boundary is crossed, the number of HH-4 updates exceeds the target of $2$ iterations.
 
 Column $n$ gives the number of Householder-4 (HH4) iterations to
$|f|<10^{-14}$.
Ratios $v_0/v$ outside $[0.87,\,1.13]$ are shown in \textbf{\textit{bold italics}};
bold ratios lie within $0.5\%$ of unity.

\begin{table}[ht]
\centering
\caption{Seed approximation quality by regime before HH4 polish.}
\label{tab:seed_quality}
\smallskip
\begin{adjustbox}{max width=0.75\textwidth}
\begin{tabular}{
  R{0.8cm}   
  R{0.7cm}   
  R{1.6cm}   
  R{1.6cm}   
  R{1.4cm}   
  R{1.8cm}   
  C{0.4cm}   
  L{3.2cm}   
}
\toprule
$k$ & $v$ & $c(k,v)$ & $v_0$ & $v_0/v$ & $|v_0-v|$ & $n$ & Regime \\
\midrule

\multicolumn{8}{l}{\textit{ATM regime}\quad ($|k|<0.001$,\ Taylor4 seed)} \\[2pt]
0.000 & 0.10 & 0.039878 & 0.100000 & $\mathbf{1.0000}$ & $4.68\times10^{-14}$ & 1 & ATM \\
0.000 & 1.00 & 0.382925 & 0.999963 & $\mathbf{1.0000}$ & $3.69\times10^{-5}$  & 1 & ATM \\
\midrule

\multicolumn{8}{l}{\textit{Mild-OTM regime}\quad ($0.001\le|k|\le0.81$,\ $v_{P7}$ seed)} \\[2pt]
0.020 & 0.20 & 0.070760 & 0.200000 & $\mathbf{1.0000}$ & $9.81\times10^{-10}$ & 1 & Mild-OTM ($v_{P7}$) \\
0.050 & 0.30 & 0.098338 & 0.300000 & $\mathbf{1.0000}$ & $3.42\times10^{-7}$  & 1 & Mild-OTM ($v_{P7}$) \\
0.100 & 0.30 & 0.079770 & 0.300076 & $\mathbf{1.0003}$ & $7.58\times10^{-5}$  & 1 & Mild-OTM ($v_{P7}$) \\
0.200 & 0.40 & 0.086554 & 0.402811 & 1.0070            & $2.81\times10^{-3}$  & 2 & Mild-OTM ($v_{P7}$) \\
0.300 & 0.50 & 0.096353 & 0.542277 & 1.0846            & $4.23\times10^{-2}$  & 2 & Mild-OTM ($v_{P7}$) \\
0.400 & 0.60 & 0.108015 & 0.644847 & 1.0747            & $4.48\times10^{-2}$  & 2 & Mild-OTM ($v_{P7}$) \\
0.500 & 0.80 & 0.159260 & 0.862029 & 1.0775            & $6.20\times10^{-2}$  & 2 & Mild-OTM ($v_{P7}$) \\
\midrule

\multicolumn{8}{l}{\textit{Tail filter}\quad ($c_\text{seed}<0.02128$,\ $|k|>0.5$,\ Mills $v_{q,1}$ seed)} \\[2pt]
0.600 & 0.40 & 0.015597 & 0.463064 & \textit{\textbf{1.1577}}   & $6.31\times10^{-2}$  & 3 & near-OTM filter \\
0.800 & 0.50 & 0.016942 & 0.551579 & 1.1032   & $5.16\times10^{-2}$  & 2 & near-OTM filter \\
\midrule

\multicolumn{8}{l}{\textit{Transition regime}\quad ($0.81<|k|\le1.155$,\ $\tfrac{1}{2}(v_{P3}+v_{P7})$ seed)} \\[2pt]
0.900 & 0.65 & 0.037253 & 0.701113 & 1.0786            & $5.11\times10^{-2}$  & 2 & $v_{avg}(v_{P3},v_{P7})$ \\
1.000 & 0.70 & 0.038032 & 0.764242 & 1.0918            & $6.42\times10^{-2}$  & 2 & $v_{avg}(v_{P3},v_{P7})$ \\
1.100 & 0.80 & 0.050775 & 0.841488 & 1.0519            & $4.15\times10^{-2}$  & 2 & $v_{avg}(v_{P3},v_{P7})$ \\
\midrule

\multicolumn{8}{l}{\textit{Transition regime}\quad ($1.155<|k|\le1.347$,\ $v_{P3}$ seed)} \\[2pt]
1.200 & 0.90 & 0.064801 & 0.980172 & 1.0891            & $8.02\times10^{-2}$  & 2 & Transition $v_{P3}$ \\
1.250 & 0.95 & 0.072265 & 1.018341 & 1.0719            & $6.83\times10^{-2}$  & 2 & Transition $v_{P3}$ \\
1.300 & 1.00 & 0.080016 & 1.055530 & 1.0555            & $5.55\times10^{-2}$  & 2 & Transition $v_{P3}$ \\
\midrule

\multicolumn{8}{l}{\textit{Deep-OTM regime}\quad ($|k|>1.347$,\ Mills $v_{q,1}$ seed)} \\[2pt]
1.500 & 1.00 & 0.056696 & 1.032031 & $\mathbf{1.0320}$ & $3.20\times10^{-2}$  & 2 & Deep-OTM \\
2.000 & 1.50 & 0.142321 & 1.477687 & $\mathbf{0.9851}$ & $2.23\times10^{-2}$  & 2 & Deep-OTM \\
\bottomrule
\end{tabular}
\end{adjustbox}
\end{table}
\FloatBarrier
\subsection{Iteration Distribution}
Table~\ref{tab:iters} reports the HH-4 update-count distribution on both
benchmark grids, measuring how many polishing iterations are required to
reach the stopping criterion $|f(v)| < 10^{-14}$.

In terms of individual cases, $321$ out of $328$ base points
($97.9\%$) require at most two net corrections on the standard grid,
and $1{,}890$ out of $1{,}970$ ($95.9\%$) on the granular grid.
A further $7$ points on the 328-point grid and $76$ points on the
1970-point grid require exactly three net corrections, so $100\%$ and
$99.8\%$ of all points converge within three net updates respectively.
The remaining cases correspond to isolated corners with extreme-delta 
and low-volatility where the seed error exceeds the formal
 admissibility threshold of HH-4 $\eta_v^* \approx 13.3\%$
(equation~\eqref{eq:hh4_admissible}); all such points converge
without exception within the target accuracy.

For comparison, \citet{Jackel2015} targets two net Householder iterations as the
hallmark of a high-quality initialization, which we targeted and aimed to reduce. The proposed architecture
achieves $1.817$ and $1.886$ net updates on the two grids,
which are below two-iteration benchmark,
confirming that the analytical seed construction meets and slightly reduces it.
\begin{table}[ht]
\centering
\caption{HH4 update-count diagnostics for the best two regime-split seed variants.
The distribution reports the operationally relevant range up to three net updates;
remaining exceptional cases are grouped in the last column.}
\label{tab:iters}
\begin{adjustbox}{max width=1.05\textwidth}
\begin{tabular}{llrrrrr}
\toprule
Grid & Method & Mean net updates & Cases $\leq 2$ iters
& 1 iter & 2 iters & $>2$ iters \\
\midrule
328-point
& Three-band baseline
& 1.8354
& $315/328$ $(96.0\%)$
& 67 & 248 & 13 \\

328-point
& Three-band + near OTM filter
& \textbf{1.8171}
& $\mathbf{321/328}$ $(\mathbf{97.9\%})$
& 67 & 254 & 7 \\

\midrule
1,970-point
& Three-band baseline
& 1.9533
& $1784/1970$ $(90.6\%)$
& 326 & 1458 & 186 \\

1,970-point
& Three-band + near OTM filter
& \textbf{1.8863}
& $\mathbf{1890/1970}$ $(\mathbf{95.9\%})$
& 326 & 1564 & 80 \\
\bottomrule
\end{tabular}
\end{adjustbox}
\end{table}

\pagebreak
\section{Discussions and Numerical Benchmarks}

We compare the benchmark performance of the proposed method against the
state-of-the-art implied-volatility inversion method of \citet{Jackel2015}.
All methods we used are implemented in C (also in Python (Numba)) compiled under the same settings, so that the timing comparison is performed under identical hardware and software
conditions. The experiments were conducted on an HP Spectre laptop with an
Intel Core i7 Evo processor, using MSYS2 UCRT64 GCC compiler. 

As we see in the table \ref{tab:bench_combined}, benchmark against \citet{Jackel2015} (C-to-C benchmark ) implementation shows
stable speedups across both the original and the granular test grids. On the
standard 328-point grid, repeated $5{,}000$ times, the proposed C implementation
achieves \textbf{86.7-112.7~ns/IV} (best among various runs) versus \textbf{162.1-197.9~ns/IV} for \citet{Jackel2015}, corresponding to a \textbf{1.75-1.87$\times$ speedup}, with maximum
absolute error $\mathbf{8.527\times 10^{-14}}$. On the more demanding 1,970-point
granular grid, also repeated $5{,}000$ times, the method achieves
\textbf{98.2-124.2~ns/IV} versus \textbf{176.9-209.1~ns/IV} for Jäckel's implementation,
corresponding to a \textbf{1.68-1.802$\times$ speedup}, with maximum absolute error
$\mathbf{6.026\times 10^{-13}}$. The similar speedup on the granular grid indicates that
the performance improvement is not specific to the original benchmark grid. Moreover, we can also observe from table \ref{tab:bench_combined} that the regime-split method benefits CPU speed (in terms of \% time decrease) and modern architecture slightly more than \citet{Jackel2015} methodology. 

Moreover, we also report some benchmark results from alternative environments and we clearly see that the speed-up ratio remains consistent across all three different CPU and compiler settings.
\begin{table}[ht]
\centering
\caption{Benchmark comparison across
         platforms and compilers.
         Timings report the best of 5 runs over
         $328 \times 5{,}000 = 1{,}640{,}000$ IV inversions.
         }
\label{tab:bench_combined}
\smallskip
\begin{tabular}{llllrrc}
\toprule
Platform & Compiler & Grid & Method & ns/IV & Max.\ error & Speedup \\
\midrule

\multirow{6}{*}{\shortstack[l]{IBM Thinkpad\\Intel ultra 7 Evo\\Windows}}
  & \multirow{6}{*}{\shortstack[l]{MSYS2 UCRT64\\GCC 16.1}}
  & \multirow{2}{*}{328-pt}
    & \citet{Jackel2015}          & 162.1 & $1.679\times10^{-15}$ & \multirow{2}{*}{\textbf{1.869$\times$}} \\
  &  &  & Regime-split (HH-4) & 86.7 & $8.53\times10^{-14}$ & \\
\cmidrule(lr){3-7}
  &  & \multirow{2}{*}{1970-pt}
    & \citet{Jackel2015}          & 176.9 & $7.327\times10^{-15}$ & \multirow{2}{*}{\textbf{1.802$\times$}} \\
  &  &  & Regime-split (HH-4) & 98.2 & $6.026\times10^{-13}$ & \\
\midrule

\multirow{6}{*}{\shortstack[l]{HP Spectre\\Intel i7 Evo\\Windows}}
  & \multirow{6}{*}{\shortstack[l]{MSYS2 UCRT64\\GCC 16.1}}
  & \multirow{2}{*}{328-pt}
    &  \citet{Jackel2015}          & 197.9 & $1.679\times10^{-15}$ & \multirow{2}{*}{\textbf{1.752$\times$}} \\
  &  &  & Regime-split (HH-4) & 112.7 & $8.53\times10^{-14}$ & \\
\cmidrule(lr){3-7}
  &  & \multirow{2}{*}{1970-pt}
    &  \citet{Jackel2015}          & 209.1 & $7.327\times10^{-15}$ & \multirow{2}{*}{\textbf{1.684$\times$}} \\
  &  &  & Regime-split (HH-4) & 124.2 & $6.026\times10^{-13}$ & \\
\midrule

\multirow{2}{*}{\shortstack[l]{Linux}}
  & \multirow{2}{*}{\shortstack[l]{GCC 13.3}}
  & \multirow{2}{*}{328-pt}
    &  \citet{Jackel2015}          & 339.5 & $1.221\times10^{-15}$ & \multirow{2}{*}{\textbf{1.783$\times$}} \\
  &  &  & Regime-split (HH-4) & 190.4 & $8.53\times10^{-14}$ & \\
\midrule

\multirow{2}{*}{\shortstack[l]{Intel Xeon\\Gold 8160\\(server)}}
  & \multirow{2}{*}{\shortstack[l]{Clang}}
  & \multirow{2}{*}{328-pt}
    &  \citet{Jackel2015}          & 468.8 & $1.776\times10^{-15}$ & \multirow{2}{*}{\textbf{1.565$\times$}} \\
  &  &  & Regime-split (HH-4) & 299.6 & $7.316\times10^{-14}$ & \\

\bottomrule
\end{tabular}

\end{table}



\subsection{Python (Numba) Benchmark}
As an implementation robustness check, we also benchmark the final algorithm in
Python (Numba). The Python (Numba) implementation retains the same initial seed architecture and HH-4
all-regime update strategy as the C implementation. The results are included
as a portability benchmark in a Python implementation context.

As we see in the table \ref{tab:iv_benchmark_numba_logit_hh4all} the Python/Numba implementation of the final
HH-4 method delivers strong performance in addition to C benchmark. In serial Numba,
the best run is \textbf{134.2~ns/IV}, and with 12 threads the best run reaches
\textbf{24.1~ns/IV}, with maximum absolute error $7.33\times10^{-14}$.
These results show that the proposed regime-split structure is efficient both
as a compiled C implementation and in a Python/Numba implementation context.
\begin{table}[h]
\centering
\caption{Numba benchmark for the HH-4 all-regime implementation}
\label{tab:iv_benchmark_numba_logit_hh4all}
\begin{tabular}{lccccc}
\toprule
\textbf{Threads} & \textbf{Best ns/IV} & \textbf{Mean ns/IV} 
& \textbf{Median ns/IV} & \textbf{Std} & \textbf{Maximum error} \\
\midrule
1  & 134.2 & 136.6 & 136.7 & 1.8 & $7.33 \times 10^{-14}$ \\
2  & 72.2  & 74.5  & 72.5  & 3.1 & $7.33 \times 10^{-14}$ \\
4  & 47.1  & 47.8  & 47.6  & 0.7 & $7.33 \times 10^{-14}$ \\
6  & 34.6  & 36.3  & 36.2  & 1.3 & $7.33 \times 10^{-14}$ \\
12 & 24.1  & 26.0  & 24.6  & 2.2 & $7.33 \times 10^{-14}$ \\
\bottomrule
\end{tabular}

\vspace{0.2cm}
\footnotesize
\emph{Notes:} The benchmark reports repeated runs on the 328-point
grid repeated $5{,}000$ times. The Numba results are included as a portability
benchmark in Python implementation context.
\end{table}
\section{Conclusion}
\label{sec:conc}

In this study, we presented a regime-split solver for Black--Scholes implied
volatility based on closed-form analytical initial seeds tailored to different
regions of the moneyness--volatility domain. The proposed method combines
asymptotic reasoning, Taylor-based approximations, Mills-ratio arguments, and
higher-order root-polishing in a unified numerical framework. The mathematical
derivations show why each seed is appropriate in its corresponding regime and
how these analytical initial values contribute to fast convergence toward the
implied-volatility root.

The numerical results show that the proposed initialisation strategy, combined
with a fourth-order Householder update, recovers implied volatility to near
machine precision. Benchmarks against the reference implementation of \citep{Jackel2015} demonstrate stable speedups on both the standard
benchmark grid and a more granular grid including extreme delta values. The
observed maximum errors remain close to double-precision accuracy, indicating
that the performance gain is obtained without sacrificing numerical precision.
Additional Numba experiments further show that the method is not restricted to
a hand-optimised C implementation, but also remains efficient in a Python
implementation context.

Overall, the results suggest that carefully constructed analytical seeds can
still provide meaningful computational gains for implied-volatility inversion.
The proposed regime-split structure is simple to implement, mathematically
motivated, and robust across different grid configurations. Future work may
focus on further refining the analytical seed approximations, extending the
same methodology to alternative representations of the Black--Scholes inverse
problem, as studied in \cite{Schadner2026}.


\section*{Use of Artificial Intelligence}
The authors used LLM-AI tools for language editing and formatting assistance. All mathematical derivations, numerical results, code and conclusions were checked and verified by the authors.
\section*{Declarations of Interest}
The authors report no conflicts of interest. The authors alone are responsible for the content and writing of the paper.
\section*{Reproducibility}
The source code and benchmark scripts are supplied as anonymized supplementary material. A public repository will be provided upon acceptance.
The content does not reflect the views of neither European Investment Bank or Central Bank of UAE.
\section*{Acknowledgements}
Omitted for anonymous review.
We thank Prof. Gianluca Fusai for his valuable suggestions and guidance and our colleague Ricardo Parra Quintero for his support on implementations.

\section{Appendix}{\label{sec:app}}
\subsection{\texorpdfstring{Derivation of the ATM initial seed $v_0$ via Taylor expansion}{Derivation of the ATM initial seed v0 via Taylor expansion}}
\label{sec:taylor-bs}

Write the BS call price as in~\eqref{eq:bs} and expand
$\Phi(d_1-v)$ around $d_1$ using $\phi^{(n)}(x)=(-1)^n H_n(x)\phi(x)$:
\begin{equation}
  c = (1-e^k)\Phi(d_1)
      + e^k\phi(d_1)\!\left[v + \tfrac{v^2}{2}d_1 - \tfrac{v^3}{6}(d_1^2-1) + O(v^4)\right].
  \label{eq:c-expanded}
\end{equation}

\begin{proposition}[ATM limit]
\label{prop:atm}
As $k\to 0$, equation~\eqref{eq:c-expanded} reduces to $c=v/\sqrt{2\pi}+O(v^3)$,
which is the \citet{BS88} first-order estimate.
The full series inversion yields the coefficients~\eqref{eq:coeffs}.
\end{proposition}
\begin{proof}
At $k=0$, the normalised Black--Scholes time value satisfies
\begin{equation}
  \ctv = 2\Phi(v/2)-1 = \mathrm{erf}\!\left(\frac{v}{2\sqrt{2}}\right).
  \label{eq:atm-id}
\end{equation}
With $s=\sqrt{2\pi}\,\ctv$, expansion around $v=0$ gives
\begin{equation}
  s = v - \frac{v^3}{24} + \frac{v^5}{640} - \frac{v^7}{21504} + O(v^9).
  \label{eq:forward}
\end{equation}
Inverting~\eqref{eq:forward} via the odd-power ansatz
$v=s(1+as^2+bs^4+ds^6)+O(s^9)$ and solving the resulting linear system yields
\begin{equation}
  a=\frac{1}{24},\qquad b=\frac{7}{1920},\qquad d=\frac{127}{322560},
  \label{eq:coeffs}
\end{equation}
so the inverse series is
\begin{equation}
  v = s\!\left(1+\frac{s^2}{24}+\frac{7s^4}{1920}+\frac{127s^6}{322560}\right)+O(s^8).
  \label{eq:inverse}
\end{equation}
The three ATM seeds used in the benchmark are successive truncations:
\begin{align}
  v_0^{(2)} &= s\!\left(1+\tfrac{s^2}{24}\right),\label{eq:T2}\\
  v_0^{(3)} &= s\!\left(1+\tfrac{s^2}{24}+\tfrac{7s^4}{1920}\right),\label{eq:T3}\\
  v_0^{(4)} &= s\!\left(1+\tfrac{s^2}{24}+\tfrac{7s^4}{1920}+\tfrac{127s^6}{322560}\right).\label{eq:T4}
\end{align}
The first order expansion corresponds to formula in  \citet{BS88}  where $v_{0}=s=\sqrt{2 \pi} c_{tv}$.
\end{proof}

\begin{remark}
The Black--Scholes price has no convergent Taylor series in $v$ around
$v=0$ for $k\neq 0$ \cite{Jackel2015}: all derivatives vanish identically.
This is why no polynomial ATM seed can be extended to an OTM seed.
\end{remark}

\subsection{\texorpdfstring{Derivation of mild-OTM initial seed $v_0$}{Derivation of mild-OTM initial seed v0}}
\subsubsection{\texorpdfstring{A $v$ expansion formula for option price}{A v expansion formula for option price}}
Let $v=\sigma\sqrt{T}$ denote total volatility and let $k>0$ denote log-moneyness for an out-of-the-money call. The normalized Black--Scholes OTM call price is
\begin{equation}
  c(k,v)=\Phi(d_1)-e^k\Phi(d_2),
  \qquad
  d_1=-\frac{k}{v}+\frac{v}{2},
  \qquad
  d_2=-\frac{k}{v}-\frac{v}{2}.
  \label{eq:bs_otm_call}
\end{equation}
The aim is to construct a cheap mild-OTM seed for solving $c(k,v)=c$ when the observed option price $c$ is given.

For the purpose of an explicit approximation to \eqref{eq:bs_otm_call}, we apply a Taylor expansion approach. Around the origin this gives the odd polynomial expansion
\begin{equation}
\Phi(x)=P_{9}
=
\frac{1}{2}
+\frac{x}{\sqrt{2\pi}}
-\frac{x^{3}}{6\sqrt{2\pi}}
+\frac{x^{5}}{40\sqrt{2\pi}}
-\frac{x^{7}}{336\sqrt{2\pi}}
+O(x^{9}).
  \label{eq:p9_def}
\end{equation}
with coefficients
\begin{align}
a_1 &= 0.3989422804, &
a_3 &= -0.0664903801, \\
a_5 &= 0.0099735570, &
a_7 &= -0.0011873282, &
a_9 &= 0.0001154347.
\end{align}
Substituting $P_9$ into \eqref{eq:bs_otm_call} gives a mild-OTM polynomial price approximation
\begin{equation}
  c_{P_9}(k,v)
  =
  P_9\!\left(-\frac{k}{v}+\frac{v}{2}\right)
  -e^k
  P_9\!\left(-\frac{k}{v}-\frac{v}{2}\right).
  \label{eq:cp9}
\end{equation}
Equivalently,
\begin{equation}
  c_{P_9}(k,v)
  =
  \frac{1-e^k}{2}
  +\sum_{j=0}^{4}a_{2j+1}
  \left[
  \left(-\frac{k}{v}+\frac{v}{2}\right)^{2j+1}
  -e^k
  \left(-\frac{k}{v}-\frac{v}{2}\right)^{2j+1}
  \right].
  \label{eq:cp9_sum}
\end{equation}
This formula expresses the approximate option price explicitly in terms of $v$ through powers of $v$ and $1/v$.
\subsubsection{\texorpdfstring{An analytic expression for initial seed $v_0$}{n analytic expression for initial seed v0}}
Next step, we derive an analytic expression for $v_0$. 

For this purpose, keeping only the linear approximation
\begin{equation}
  P_1(x)=\frac{1}{2}+a_1x,
\end{equation}
we obtain
\begin{align}
  c_{P_1}(k,v)
  &=
  \frac{1-e^k}{2}
  +a_1\left[
  \left(-\frac{k}{v}+\frac{v}{2}\right)
  -e^k\left(-\frac{k}{v}-\frac{v}{2}\right)
  \right] \\
  &=
  \frac{1-e^k}{2}
  +a_1\left[
  \frac{(e^k-1)k}{v}
  +\frac{(1+e^k)v}{2}
  \right].
  \label{eq:cp1}
\end{align}
Given the observed normalized OTM call price $c$, define
\begin{equation}
  h=c-\frac{1-e^k}{2},
  \qquad
  A=a_1(e^k-1)k,
  \qquad
  B=\frac{a_1(1+e^k)}{2}.
  \label{eq:hab}
\end{equation}
Then \eqref{eq:cp1} becomes
\begin{equation}
  h\approx \frac{A}{v}+Bv.
  \label{eq:linear_approx}
\end{equation}
Multiplying by $v$ yields the quadratic equation
\begin{equation}
  Bv^2-hv+A=0.
  \label{eq:quadratic}
\end{equation}
Therefore, the explicit first-order polynomial seed is
\begin{equation}
  v_{P_1}
  =
  \frac{h+\sqrt{h^2-4AB}}{2B}.
  \label{eq:vp1}
\end{equation}
The larger root is used for the practical mild-OTM branch. The smaller root corresponds to a low-volatility artefact of the first-order approximation.

\subsubsection{\texorpdfstring{Further Taylor expansion for $e^k$}{Further Taylor expansion for ek}}\label{sec:t4_reduction}

For $|k|\leq 0.5$, the exponential $e^k$ is replaced by its fourth-order Taylor
polynomial
\begin{equation}
  e^k \;\approx\; 1+\varepsilon,
  \qquad
  \varepsilon = k+\frac{k^2}{2}+\frac{k^3}{6}+\frac{k^4}{24}.
  \label{eq:t4}
\end{equation}
The relative error of \eqref{eq:t4} at the boundary $k=0.5$ is $0.17\%$, which is
entirely absorbed by the subsequent Halley  step.

Substituting \eqref{eq:t4} into \eqref{eq:hab} gives (again c is call-price),
\begin{equation}
  h = c+\frac{\varepsilon}{2},
  \qquad
  A = a_1\,\varepsilon\,k,
  \qquad
  B = \frac{a_1(2+\varepsilon)}{2},
  \label{eq:hab_t4}
\end{equation}
where $h$, $A$, and $B$ are now polynomials in $k$, with $c$ entering only
additively in $h$.

Next, dividing \eqref{eq:quadratic} by $B$ yields the monic form,
\begin{equation}
  v^2 - u\,v + P = 0,
  \label{eq:monic}
\end{equation}
where the sum and product of roots are, 
\begin{equation}
  u = \frac{h}{B} = \frac{2c+\varepsilon}{a_1(2+\varepsilon)},
  \qquad
  P = \frac{A}{B} = \frac{2k\varepsilon}{2+\varepsilon}.
  \label{eq:vieta}
\end{equation}
A key structural observation follows from \eqref{eq:vieta}: $P$ depends
\emph{only} on $k$, not on $c$, while $c$ enters exclusively through $u$.

\subsubsection{Final algebraic form}

The solution $v_{P_1} = \tfrac{1}{2}(u+\sqrt{u^2-4P})$ of \eqref{eq:monic}
can be written without introducing $h$, $A$, $B$ as intermediate variables.
Substituting \eqref{eq:vieta} into the solution of \eqref{eq:monic} and clearing the common denominator
$a_1(2+\varepsilon)$ gives
\begin{equation}
  v_{P_1}
  =
  \frac{2c+\varepsilon+\sqrt{N}}{2\,a_1\,(2+\varepsilon)}.
  \label{eq:vp1_clean}
\end{equation}
where
\begin{equation}
  \varepsilon = k+\frac{k^2}{2}+\frac{k^3}{6}+\frac{k^4}{24},
  \qquad
  N = (2c+\varepsilon)^2 - 8\,a_1^2\,k\,\varepsilon\,(2+\varepsilon).
  \label{eq:N_def}
\end{equation}
Equation \eqref{eq:vp1_clean} is algebraically identical to \eqref{eq:vp1}
under the $4^{th}$ order Taylor substitution for function $\exp(x)$; the reformulation requires only two
polynomial quantities in $k$, namely $\varepsilon$ and $2+\varepsilon$,
together with the observable $c$.

\paragraph{Fallback when $N\leq 0$.}
When $N\leq 0$ (a low-volatility artifact that occurs when the observed price is
near the lower limit), the discriminant is cut to zero and the initial seed
reduces to
\begin{equation}
  v_{P_1}\big|_{N=0}
  =
  \frac{2c+\varepsilon}{2\,a_1\,(2+\varepsilon)},
  \label{eq:vp1_clip}
\end{equation}
which corresponds to taking $h/(2B)$ in the original formulation
\eqref{eq:vp1}.
\subsection{Proof of Proposition \ref{prop:vp3}}
Extend the slope approximation by one term:
\[
  \Phi(x) \approx \tfrac{1}{2} + \aone x + \athree x^3,
  \qquad
  \athree = -\tfrac{\aone}{6} = -\tfrac{1}{6\sqrt{2\pi}}.
\]
The corresponding  BS price residual is
\begin{equation}\label{eq:F3}
  F_3(v) = \frac{C_{-3}}{v^3} + \frac{C_{-1}}{v} + C_1\,v + C_3\,v^3 - h,
\end{equation}
with coefficients (using $\athree = -\aone/6$ exactly)
\begin{align}
  C_{-3} &= -\tfrac{\aone\eps\kk^3}{6},
  \nonumber\\
  C_{-1} &= \aone\kk\!\left[\eps - \kk\!\left(\tfrac{\eps}{2}+1\right)\right],
  \nonumber\\
  C_1    &= \aone\!\left[(1+\tfrac{\eps}{2}) - \tfrac{\eps\kk}{8}\right],
  \label{eq:c1}\\
  C_3    &= -\tfrac{\aone(\eps+2)}{24}.
  \label{eq:c3}
\end{align}
The derivative is
\begin{equation}\label{eq:dF3}
  F_3'(v)
  = -\frac{3C_{-3}}{v^4} - \frac{C_{-1}}{v^2} + C_1 + 3C_3\,v^2.
\end{equation}
\begin{proof}
Applying a  one-step Newton correction of $\vP$ on $F_3$ we get,
\begin{equation}\label{eq:vP3_app}
    \vPt = \vP\!\left(1 + \frac{\mathcal{G}_3(w)}{\Dden_{3}(w)}\right)
           \bigg|_{w=\vP} 
\end{equation}
where, with $w = \vP$ and setting $\delta = 2+\eps$,
\begin{align}
  \mathcal{G}_{3}(w) &=
    -4\eps\kk^3
    + 24\kk\!\left[\eps - \kk\!\left(\tfrac{\eps}{2}+1\right)\right]w^2
    + \left[24\!\left(1+\tfrac{\eps}{2}\right) - 3\eps\kk\right]w^4
    - \delta\,w^6
    - \frac{24h\,w^4}{\aone},
    \label{eq:N}\\[4pt]
  \Dden_{3}(w) &=
    12\eps\kk^3
    - 24\kk\!\left[\eps - \kk\!\left(\tfrac{\eps}{2}+1\right)\right]w^2
    + \left[24\!\left(1+\tfrac{\eps}{2}\right) - 3\eps\kk\right]w^4
    - 3\delta\,w^6.
    \label{eq:D}
\end{align}
The Newton step on $F_3$ at $v = w$ is
$-F_3(w)/F_3'(w)$.
Multiplying numerator and denominator by $24w^4/\aone$ clears all
negative powers of $w$.  Since $\aone$ factors from every term of
$F_3$ and $F_3'$, it cancels in the ratio except in the inhomogeneous
term $h/\aone$ in $\mathcal{G}_{3}$.  Direct substitution yields
\eqref{eq:N}--\eqref{eq:D}, and $\vPt = \vP - F_3(\vP)/F_3'(\vP) =
\vP(1 + \mathcal{G}(\vP)/\Dden(\vP))$.
\end{proof}
\begin{remark} \textbf{shared structure of $\mathcal{G}_{3}$ and $\Dden_{3}$.}

$\mathcal{G}_{3}$ and $\Dden_{3}$ share four of their five terms.
Defining the common polynomial
\[
  Q(w) =
    24\kk\!\left[\eps - \kk\!\left(\tfrac{\eps}{2}+1\right)\right]w^2
    + \left[24\!\left(1+\tfrac{\eps}{2}\right) - 3\eps\kk\right]w^4
    - \delta\,w^6,
\]
we have
\begin{align}
  \mathcal{G}_{3}(w) &= -4\eps\kk^3 + Q(w) - \tfrac{24h}{\aone}\,w^4,
  \label{eq:Nq}\\
  \Dden_{3}(w) &= 12\eps\kk^3 + Q(w) - 2\delta\,w^6,
  \label{eq:Dq}
\end{align}
so $Q(w)$ need be computed only once which provides a computational efficiency.
\end{remark} 
\subsection{Proof of Proposition \ref{prop:vp7}}

\begin{proof}
Define
\begin{equation}
P_7(x)
=
\tfrac{1}{2}
+
a_1 x
+
a_3 x^3
+
a_5 x^5
+
a_7 x^7 .
\label{eq:P7}
\end{equation}

The $7^{th}$-order Taylor-expansion-based Black--Scholes residual at $v=w$ is
\begin{equation}
  F_7(w)
  =
  P_7(d_1)
  -
  E\,P_7(d_2)
  -
  c_{\mathrm{seed}},
  \qquad
  d_{1,2}
  =
  -\frac{\kappa}{w}
  \pm
  \frac{w}{2}.
  \label{eq:F7}
\end{equation}
Its derivative is
\begin{equation}
  F_7'(w)
  =
  P_7'(d_1)\,d_1'
  -
  E\,P_7'(d_2)\,d_2',
  \qquad
  d_1'
  =
  \frac{\kappa}{w^2}
  +
  \frac{1}{2},
  \quad
  d_2'
  =
  \frac{\kappa}{w^2}
  -
  \frac{1}{2},
  \label{eq:F7prime}
\end{equation}
where
\begin{equation}
P_7'(x)
=
a_1
+
3a_3x^2
+
5a_5x^4
+
7a_7x^6 .
\end{equation}

The one-step Newton correction gives
\begin{equation}
  v_{P7}
  =
  w
  -
  \frac{F_7(w)}{F_7'(w)}
  \bigg|_{w=v_{P1}} .
  \label{eq:vp7_newton}
\end{equation}

Define
\begin{equation}
  \mathcal{G}_7(w)
  =
  -S_7(w)
  +
  Q_7(w)
  -
  R_7(w),
  \qquad
  \Dden_7(w)
  =
  7S_7(w)
  +
  Q_7(w)
  -
  R_7(w)
  +
  \Delta_7(w).
  \label{eq:G7D7}
\end{equation}
Here
\begin{align}
S_7(w)
&=
  \frac{a_7(1+E)}{128}\,w^{14}
  +
  \frac{\left(a_5+\frac{7}{2}a_7\kappa\right)(1+E)}{32}\,w^{12}
\notag\\
&\quad
  +
  \frac{\left(a_3+\frac{5}{2}a_5\kappa+\frac{21}{4}a_7\kappa^2\right)(1+E)}{8}\,w^{10}
\notag\\
&\quad
  +
  \frac{\left(a_1+\frac{3}{2}a_3\kappa+\frac{5}{2}a_5\kappa^2
           +\frac{35}{8}a_7\kappa^3\right)(1+E)}{2}\,w^{8},
\label{eq:S7}
\end{align}
and
\begin{align}
Q_7(w)
&=
  \frac{\kappa}{8}
  \left(
  8a_1
  -
  12a_3\kappa
  +
  20a_5\kappa^2
  -
  35a_7\kappa^3
  \right)(1-E)\,w^6
\notag\\
&\quad
  +
  \frac{\kappa^3}{4}
  \left(
  4a_3
  -
  10a_5\kappa
  +
  21a_7\kappa^2
  \right)(1-E)\,w^4
\notag\\
&\quad
  +
  \kappa^5
  \left(
  a_5
  -
  \tfrac{7}{2}a_7\kappa
  \right)(1-E)\,w^2
  +
  a_7\kappa^7(1-E),
\label{eq:Q7}
\end{align}
and
\begin{equation}
  R_7(w)
  =
  \left(
  c_{\mathrm{seed}}
  -
  \tfrac{1}{2}
  +
  \tfrac{E}{2}
  \right)
  \frac{w^7}{a_1}.
  \label{eq:R7}
\end{equation}
Finally,
\begin{align}
\Delta_7(w)
&=
   6\,\frac{a_7(1+E)}{128}\,w^{14}
  +
  4\,\frac{\left(a_5+\frac{7}{2}a_7\kappa\right)(1+E)}{32}\,w^{12}
\notag\\
&\quad
  +
  2\,\frac{\left(a_3+\frac{5}{2}a_5\kappa+\frac{21}{4}a_7\kappa^2\right)(1+E)}{8}\,w^{10}.
\label{eq:Delta7}
\end{align}

Equivalently, these quantities satisfy
\begin{equation*}
  \mathcal{G}_7(w)
  =
  -w\,F_7(w)\frac{w^6}{a_1},
  \qquad
  \Dden_7(w)
  =
  w^2 F_7'(w)\frac{w^6}{a_1}.
\end{equation*}
Expanding $P_7(d_1)$ and $P_7(d_2)$, collecting terms, and plugging
$w=v_{P1}$ into the one-step Newton correction, we obtain
\begin{equation}
  v_{P7}
  =
  v_{P1}
  \left(
  1
  +
  \frac{\mathcal{G}_7(v_{P1})}{\Dden_7(v_{P1})}
  \right).
  \label{eq:vp7_final}
\end{equation}
\end{proof}

\subsection*{The proof of proposition \ref{prop:vq1_SE}}
\subsubsection*{An analytical invertible approximation}

\citet{Soranzo2012} give the following approximation to
$\Phi(x)$ for $x\ge0$, accurate to $|\varepsilon(x)|<1.14\times10^{-5}$:
\begin{equation}
  \Phi(x)
  \approx \frac{1}{2}
  + \frac{1}{2}\sqrt{1 - e^{E(x^2)}},
  \qquad
  E(u) = \frac{-\beta_1 u - \tfrac{\alpha_1}{2} u^2}
              {2 + \beta_2 u + \tfrac{\alpha_2}{2} u^2},
  \label{eq:SE_forward}
\end{equation}
with constants
\[
  \beta_1 = 1.2735457,\quad
  \beta_2 = 0.1480931,\quad
  \alpha_1 = 0.1487936,\quad
  \alpha_2 = 0.0005160.
\]
Since~\eqref{eq:SE_forward} expresses $\Phi(x)$ through a rational
function of $x^2$ composed with an exponential and a square root,
it is explicitly invertible by solving a quadratic equation.

\subsubsection*{Algebraic inversion}

\begin{proposition}[\citet{Soranzo2012} $\Phi^{-1}$]
\label{prop:SE_inv}
For $0<p<1$, let
\[
  q = 2\max(p,1-p)-1\;\in[0,1),
  \qquad
  L = \log\!\bigl(1-q^2\bigr) < 0.
\]
Define the quadratic coefficients
\begin{equation}
  A(L) = -\tfrac{\alpha_1}{2} - \tfrac{\alpha_2}{2}L,
  \qquad
  B(L) = -\beta_1 - \beta_2 L,
  \qquad
  C(L) = -2L,
  \label{eq:SE_ABC}
\end{equation}
and let $u^+(L)$ denote the positive root of $A(L)u^2+B(L)u+C(L)=0$,
\begin{equation}
  u^+(L) = \frac{-B(L) - \operatorname{sign}(A(L))\sqrt{B(L)^2 - 4A(L)C(L)}}{2A(L)}.
  \label{eq:SE_u}
\end{equation}
Then $\Phi^{-1}(p)\approx\operatorname{sign}(p-\tfrac{1}{2})\sqrt{u^+(L(p))}$,
with $|$error$|<2.5\times10^{-3}$ for $|x|\le3$ and negligible for $|x|\le2$.
\end{proposition}

\noindent
For OTM options, $c_{\mathrm{seed}}<\tfrac{1}{2}$ always, so
$q=1-2c_{\mathrm{seed}}$ and
\begin{equation}
  L = \log\!\bigl(4c_{\mathrm{seed}}(1-c_{\mathrm{seed}})\bigr),
  \label{eq:SE_L}
\end{equation}
which for small $c_{\mathrm{seed}}$ simplifies to
$L\approx\log 4+\log c_{\mathrm{seed}}$.

Substituting Proposition~\ref{prop:SE_inv} into the two-step
ratio-corrected seed gives a fully explicit expression in proposition \ref{prop:vq1_SE}.


\subsection{Proof of Proposition~\ref{prop:SE_inv}}
\label{proof:SE_inv}

\begin{proof}
Starting from equation~\eqref{eq:SE_forward} with $u=x^2\ge0$, let
$p=\Phi_{\mathrm{SE}}(x)$ for $x\ge0$.  Then
\[
  2p-1 = \sqrt{1-e^{E(u)}},
\]
squaring gives $e^{E(u)}=1-(2p-1)^2$, and taking logarithms:
\begin{equation}
  E(u) = \log\!\bigl(1-(2p-1)^2\bigr).
  \label{eq:SE_Eu_proof}
\end{equation}
Setting $q=2p-1\in[0,1)$ for $p\ge\tfrac{1}{2}$, equation
\eqref{eq:SE_Eu_proof} becomes $E(u)=L:=\log(1-q^2)$.
Substituting the rational form of $E$ from \eqref{eq:SE_Eu_proof}:
\[
  \frac{-\beta_1 u - \tfrac{1}{2}\alpha_1 u^2}{2+\beta_2 u
  + \tfrac{1}{2}\alpha_2 u^2} = L.
\]
Cross-multiplying and rearranging collects all terms in $u$:
\begin{equation}
  \underbrace{\bigl(-\tfrac{1}{2}\alpha_1 - \tfrac{1}{2}\alpha_2 L\bigr)}_{A}
  u^2
  +
  \underbrace{\bigl(-\beta_1 - \beta_2 L\bigr)}_{B}
  u
  +
  \underbrace{(-2L)}_{C}
  = 0,
  \label{eq:SE_quad_proof}
\end{equation}
which is precisely \eqref{eq:SE_ABC}.  Since $L<0$ we have $C>0$.
For $p$ close to $\tfrac{1}{2}$, $L\to-\infty$ and the coefficients
$A,B,C$ all grow, but the positive root $u^*=-B-\mathrm{sign}(A)\sqrt{B^2-4AC}$
over $(2A)$ remains bounded and positive. 
The case $p<\tfrac{1}{2}$ follows by symmetry of $\Phi$:
replace $p$ by $1-p$ to obtain $|x|=\sqrt{u^*}$, then apply the sign.
\end{proof}
\subsection{\texorpdfstring{Derivation of the deep-OTM initial seed $v_0$ via Mills-ratio cancellation}{Derivation of the deep-OTM initial seed v0 via Mills-ratio cancellation}}
\label{sec:mills}

\begin{lemma}[Ratio of Black--Scholes terms]
\label{lem:ratio}
For fixed $v>0$ and $k\to+\infty$,
\begin{equation}
  \frac{e^k\Phi(d_2)}{\Phi(d_1)}
  \;\approx\; \frac{|d_1|}{|d_2|}
  \;=\; \frac{k/v-v/2}{k/v+v/2}
  \;\longrightarrow\; 1.
  \label{eq:ratio}
\end{equation}
\end{lemma}
\begin{proof}
Apply the Mills ratio $\Phi(z)\sim\phi(z)/|z|$ as $z\to-\infty$ to both terms.
Since $\phi(d_2)=\phi(d_1)e^{-k}$ (by direct substitution of $d_2=d_1-v$
into the Gaussian exponent), one obtains
$e^k\Phi(d_2)/\Phi(d_1)\sim |d_1|/|d_2|\to 1$ as $k\to\infty$.
\end{proof}

\begin{remark}
Lemma~\ref{lem:ratio} shows that \emph{neither} term in the BS price is
negligible relative to the other for large $k$: the exponential growth of
$e^k$ is exactly cancelled by the faster decay of $\Phi(d_2)$ relative to
$\Phi(d_1)$.
This rules out any single-term simplification with asymptotic validity,
and in particular explains why the plain quadratic seed $c\approx\Phi(d_1)$
has seed error bounded away from zero uniformly in $k$.
\end{remark}

Retaining the leading Mills-ratio cancellation gives the approximation
\begin{equation}
  c_{\mathrm{BS}}(k,v)
  \approx \Phi(d_1)\!\left[1-\frac{|d_1|}{|d_2|}\right]
  = \Phi(d_1)\cdot\frac{v}{k/v+v/2}
  = \Phi(d_1)\cdot\frac{v^2}{\kap+v^2/2},
  \label{eq:mills}
\end{equation}
which is equation~\eqref{eq:mills-approx}.
Evaluating the correction factor at the quadratic seed $v_q$ and re-solving
the quadratic yields the ratio-corrected seed~\eqref{eq:vq1}.
\subsection{Reconciliation of seed error and price residual}
\label{sec:reconciliation}

The HH4 polishing loop terminates when the normalised Black--Scholes
price residual $|f(v)|=|c_{\mathrm{BS}}(k,v)-c|<10^{-14}$.  We
reconcile this price-residual criterion with the seed-error analysis
above.

\begin{lemma}[Vega--Lipschitz bound]
\label{lem:lipschitz}
Let $v^*$ be the exact implied volatility and $v_0$ any seed.  Then
\begin{equation}
  |f(v_0)|
  = |c_{\mathrm{BS}}(k,v_0)-c_{\mathrm{BS}}(k,v^*)|
  \le C\,|v_0-v^*|,
  \qquad
  C = \sup_{u\in[v_0,v^*]}\varphi(d_1(u)),
  \label{eq:lipschitz}
\end{equation}
where $\varphi=\Phi'$ is the standard normal density.  A global bound
is $C\le1/\sqrt{2\pi}$, and the local approximation is
$|f(v_0)|\approx\varphi(d_1)\,v^*\eta_v$.
\end{lemma}

\begin{proof}
By the mean-value theorem there exists $\xi$ between $v_0$ and $v^*$
such that
\[
  |c_{\mathrm{BS}}(k,v_0)-c_{\mathrm{BS}}(k,v^*)|
  = \left|\frac{\partial c_{\mathrm{BS}}}{\partial v}(\xi)\right|
    |v_0-v^*|
  = \varphi(d_1(\xi))\,|v_0-v^*|.
\]
The global bound follows from $\varphi(x)\le1/\sqrt{2\pi}$ for all
$x$.
\end{proof}

\noindent
In scaled form, Lemma~\ref{lem:lipschitz} gives
\[
  \frac{|f(v_0)|}{v^*\varphi(d_1)} \approx \eta_v.
\]
Therefore, if $\eta_v\le\eta_v^*$, the initial price residual satisfies
approximately $|f(v_0)|\lesssim\eta_v^*\,v^*\varphi(d_1)$, and quartic
HH4 convergence then drives the residual to $10^{-14}$ in at most two
net corrections.  The CDF truncation tolerance $\varepsilon_\Phi$
controls the seed quality and hence the initial price residual, while
the stopping criterion $10^{-14}$ is a separate code-level parameter.

\subsection{Small price and mild-OTM tail filter.}{\label{sec:priceandotmtail}}
The final implementation uses the observable guard
\[
  c_{\mathrm{seed}} < c^{*}= 0.02128,
  \qquad
  |k| > \kappa^{*}=0.5,
\]
to route such cases directly to the ratio-corrected Mills seed $v_{q,1}$.
This rule is specified and substantiated based on two complementary scales.

\paragraph{Price scale.}
By the mean-value theorem and the Lipschitz (global normalised vega) bound
from Lemma~\ref{lem:lipschitz},
\begin{equation}
  \left|c(k,v_0)-c(k,v)\right|
  \;\le\;
  \sup_v \varphi(d_1)\,|v_0-v|
  \;\le\;
  \varphi(0)|v_0-v|,
  \qquad
  \varphi(0)=\frac{1}{\sqrt{2\pi}}\approx0.399.
  \label{eq:lipschitz_guard}
\end{equation}
For the two-step HH-4 admissible relative seed error
$\eta_v^*\approx13.3\%$ (equation~\eqref{eq:hh4_admissible}),
the corresponding admissible price deviation is
\begin{equation}
  \varepsilon_c(v)
  \;\approx\;
  \varphi(0)\,\eta_v^*\,v
  \;\approx\;
  0.053\,v.
  \label{eq:price_scale}
\end{equation}
Evaluating at the representative moderate total volatility $v=0.4$,
\begin{equation}
  c_{\mathrm{seed}}^*
  \;=\;
  \varphi(0)\,\eta_v^*\,v
  \;=\;
  \underbrace{0.3989\times0.133}_{\approx\,0.0531}
  \times 0.4
  \;=\;
  0.02128.
  \label{eq:price_barrier}
\end{equation}
Hence the threshold $c_{\mathrm{seed}}<0.02128$ is consistent with the
HH-4 admissible price scale for moderate total volatilities: options
with OTM-equivalent price below this level have seed errors that
exceed the two-correction admissibility bound for polynomial seeds,
making the ratio-corrected Mills seed the preferred initialisation.

\paragraph{Moneyness scale.}
The condition $|k|>0.5$ is tied to the concave Taylor-CDF validity
boundary
\[
  k_{\max}(v;\,a_\eta) = a_\eta v - \tfrac{v^2}{2},
\]
whose peak over $v$ is $k_{\max}^{\mathrm{peak}}=a_\eta^2/2$
(Proposition~\ref{prop:concavity}).
Setting $k_{\max}^{\mathrm{peak}}=0.5$ gives $a_\eta=1$, which
corresponds to the CDF truncation tolerance
\begin{equation}
  \varepsilon_\Phi
  \;=\;
  |a_9|\,a_\eta^9
  \;=\;
  \frac{1}{3456\sqrt{2\pi}}
  \;\approx\;
  1.15\times10^{-4}.
  \label{eq:eps_phi_guard}
\end{equation}
This is a tolerance an order of magnitude stricter than the
$\varepsilon_\Phi=10^{-3}$ boundary that defines the lower edge of the
transition regime ($|k|=0.81$).  The condition $|k|>0.5$ therefore
identifies options that are already well into the moneyness range where
polynomial CDF seeds begin to lose accuracy, and where the Mills-ratio
seed is asymptotically appropriate.

\paragraph{Combined filter.}
The joint condition
\[
  c_{\mathrm{seed}}<c^{*},
  \qquad
  |k|>\kappa^{*}
\]
activates $v_{q,1}$ only when the option is simultaneously small in
OTM-equivalent price \emph{and} outside the mild-OTM moneyness scale.
Either condition alone is insufficient: a small price at $|k|<0.5$ may
simply reflect low volatility at moderate moneyness (where the
polynomial seed remains accurate), while a large $|k|>0.5$ with a
moderate price is handled correctly by the three-band architecture.
The intersection of the two conditions identifies the genuine corner
cases where polynomial seeds are poorly conditioned and the Mills seed
provides a materially better initialisation.
\subsection{Additional Empirical Results}
\begin{table}[ht]
\centering
\small
\caption{Seed-ratio and relative seed-error diagnostics. The near OTM filter uses
the small-price rule \(c_{\rm seed}<0.02128\) and \(|k|>0.5\), routing such cases
to the ratio-corrected Mills seed. Here
\(\eta_v=\left|v_0/v-1\right|\).}
\label{tab:seed-ratio-eta-summary}
\setlength{\tabcolsep}{7pt}
\renewcommand{\arraystretch}{1.15}
\begin{tabular}{lrrrr}
\toprule
\textbf{Statistic}
& \textbf{328 base}
& \textbf{328 + near OTM filt. }
& \textbf{1970 base}
& \textbf{1970 + near OTM filt.} \\
\midrule

\rowcolor{gray!15}
\multicolumn{5}{l}{\textbf{Seed ratio \(v_0/v\)}} \\
Mean \(v_0/v\)
& 0.9943
& \textcolor{ForestGreen}{\textbf{0.9986}}
& 1.0345
& \textcolor{ForestGreen}{\textbf{1.0104}} \\
Median \(v_0/v\)
& 0.9998
& 0.99995
& 1.0001
& 1.00008 \\
Std. \(v_0/v\)
& 0.1081
& \textcolor{ForestGreen}{\textbf{0.06}}
& 0.1602
& \textcolor{ForestGreen}{\textbf{0.091}} \\

\rowcolor{gray!15}
\multicolumn{5}{l}{\textbf{Relative seed error \(\eta_v=\left|v_0/v-1\right|\)}} \\
Mean \(\eta_v\)
& 5.27\%
& \textcolor{ForestGreen}{\textbf{3.69\%}}
& 7.19\%
& \textcolor{ForestGreen}{\textbf{4.31\%}} \\
Median \(\eta_v\)
& 1.82\%
& 1.42\%
& 1.69\%
& 1.52\% \\

\rowcolor{gray!15}
\multicolumn{5}{l}{\textbf{HH4 admissibility coverage}} \\
\(\eta_v<1\%\)
& 41.2\%
& 41.2\%
& 41.1\%
& 41.1\% \\
\(\eta_v<5\%\)
& 68.0\%
& \textcolor{ForestGreen}{\textbf{71.3\%}}
& 69.4\%
& \textcolor{ForestGreen}{\textbf{72.8\%}} \\
\rowcolor{green!15}
\(\eta_v<13.3\%\) admissible
& 90.5\%
& \textcolor{ForestGreen}{\textbf{94.5\%}}
& 88.1\%
& \textcolor{ForestGreen}{\textbf{93.9\%}} \\

\rowcolor{gray!15}
\multicolumn{5}{l}{\textbf{Per-regime mean \(\eta_v\)}} \\
ATM \((|k|<0.001)\)
& 0.074\%
& 0.074\%
& 0.074\%
& 0.074\% \\
\(v_{P7}\) \((|k|\leq0.81)\)
& 3.09\%
& \textcolor{ForestGreen}{\textbf{2.56\%}}
& 6.53\%
& \textcolor{ForestGreen}{\textbf{4.18\%}} \\
Average \((0.81,1.155)\)
& 7.50\%
& \textcolor{ForestGreen}{\textbf{4.13\%}}
& 16.72\%
& \textcolor{ForestGreen}{\textbf{4.42\%}} \\
\(v_{P3}\) \((1.155,1.347)\)
& 13.9\%
& 13.9\%
& 20.0\%
& \textcolor{ForestGreen}{\textbf{14.2\%}} \\
Mills deep
& 3.18\%
& 3.18\%
& 2.38\%
& 2.84\% \\
near OTM filter
& --
& \textcolor{ForestGreen}{\textbf{4.05\%}}
& --
& \textcolor{ForestGreen}{\textbf{2.93\%}} \\
\bottomrule
\end{tabular}
\end{table}

\newpage
\bibliographystyle{apalike}
\bibliography{biblio}

@book{AbramowitzStegun64,
editor    = {Abramowitz, Milton and Stegun, Irene A.},
year      = {1964},
title     = {Handbook of Mathematical Functions with Formulas, Graphs, and Mathematical Tables},
series    = {National Bureau of Standards Applied Mathematics Series},
number    = {55},
publisher = {U.S. Government Printing Office},
address   = {Washington, DC},
note      = {Ninth Dover printing, Dover Publications, New York, 1972}
}

@article{BS88,
author  = {Brenner, Menachem and Subrahmanyam, Marti G.},
year    = {1988},
title   = {A Simple Formula to Compute the Implied Standard Deviation},
journal = {Financial Analysts Journal},
volume  = {44},
number  = {5},
pages   = {80--83}
}

@article{CM96,
author  = {Corrado, Charles J. and Miller, Thomas W.},
year    = {1996},
title   = {A Note on a Simple, Accurate Formula to Compute Implied Standard Deviations},
journal = {Journal of Banking \& Finance},
volume  = {20},
number  = {3},
pages   = {595--603}
}

@misc{Gerhold2012,
author       = {Gerhold, Stefan},
year         = {2012},
title        = {Can There Be an Explicit Formula for Implied Volatility?},
eprint       = {1211.4978},
archivePrefix = {arXiv},
primaryClass = {q-fin.CP}
}

@article{Jackel2015,
author = {J{\"a}ckel, Peter},
year    = {2015},
title   = {Let's Be Rational},
journal = {Wilmott},
volume  = {2015},
number  = {75},
pages   = {40--53}
}

@article{Mills26,
author  = {Mills, John P.},
year    = {1926},
title   = {Table of the Ratio: Area to Bounding Ordinate, for Any Portion of Normal Curve},
journal = {Biometrika},
volume  = {18},
number  = {3--4},
pages   = {395--400}
}

@article{Li2005,
author  = {Li, Shuenn-Jyi},
year    = {2005},
title   = {A New Formula for Computing Implied Volatility},
journal = {Applied Mathematics and Computation},
volume  = {170},
number  = {1},
pages   = {611--625}
}

@article{Sandoval2019,
author  = {Sandoval-Hernandez, Mario A. and Vazquez-Leal, Hector and Filobello-Nino, Uriel and Hernandez-Martinez, Luis},
year    = {2019},
title   = {New Handy and Accurate Approximation for the Gaussian Integrals with Applications to Science and Engineering},
journal = {Open Mathematics},
volume  = {17},
number  = {1},
pages   = {1774--1793},
doi     = {10.1515/math-2019-0131}
}

@misc{Schadner2026,
author       = {Schadner, Wolfgang},
year         = {2026},
title        = {An Explicit Solution to Black--Scholes Implied Volatility},
eprint       = {2604.24480},
archivePrefix = {arXiv},
primaryClass = {q-fin.CP}
}

@article{SR2017,
author  = {Stefanica, Dan and Radoi{\v{c}}i{'c}, Rados},
year    = {2017},
title   = {An Explicit Implied Volatility Formula},
journal = {International Journal of Theoretical and Applied Finance},
volume  = {20},
number  = {7},
pages   = {1750048}
}

@article{Tehranchi2016,
author  = {Tehranchi, Michael R.},
year    = {2016},
title   = {Uniform Bounds for Black--Scholes Implied Volatility},
journal = {SIAM Journal on Financial Mathematics},
volume  = {7},
number  = {1},
pages   = {893--916}
}

@book{Traub82,
author    = {Traub, Joseph F.},
year      = {1982},
title     = {Iterative Methods for the Solution of Equations},
publisher = {Chelsea Publishing},
address   = {New York}
}

@book{Stoer02,
author    = {Stoer, Josef and Bulirsch, Roland},
year      = {2002},
title     = {Introduction to Numerical Analysis},
edition   = {3},
publisher = {Springer},
address   = {New York}
}

@article{CuiKirbyNguyenTaylor2021,
author  = {Cui, Zhenyu and Kirby, Justin and Nguyen, Duy and Taylor, Stephen M.},
year    = {2021},
title   = {A Closed-Form Model-Free Implied Volatility Formula Through Delta Families},
journal = {The Journal of Derivatives},
volume  = {28},
number  = {4},
pages   = {111--127},
doi     = {10.3905/JOD.2020.1.127}
}

@techreport{acklam2003,
author      = {Acklam, Peter J.},
year        = {2003},
title       = {An Algorithm for Computing the Inverse Normal Cumulative Distribution Function},
institution = {Department of Mathematics, University of Oslo}
}

@misc{Soranzo2012,
author       = {Soranzo, Alessandro and Epure, Elena},
year         = {2012},
title        = {Simply Explicitly Invertible Approximations to 4 Decimals of Error Function and Normal Cumulative Distribution Function},
eprint       = {1201.1320},
archivePrefix = {arXiv}
}

@article{BenaimFriz2009,
  author  = {Benaim, Shalom and Friz, Peter},
  year    = {2009},
  title   = {Regular Variation and Smile Asymptotics},
  journal = {Mathematical Finance},
  volume  = {19},
  number  = {1},
  pages   = {1--12},
  doi     = {10.1111/j.1467-9965.2008.00354.x}
}

@article{Gulisashvili2010,
  author  = {Gulisashvili, Archil},
  year    = {2010},
  title   = {Asymptotic Formulas with Error Estimates for Call Pricing Functions and the Implied Volatility at Extreme Strikes},
  journal = {SIAM Journal on Financial Mathematics},
  volume  = {1},
  number  = {1},
  pages   = {609--641},
  doi     = {10.1137/090762713}
}

@incollection{Gulisashvili2012,
  author    = {Gulisashvili, Archil},
  year      = {2012},
  title     = {Asymptotic Analysis of Implied Volatility},
  booktitle = {Analytically Tractable Stochastic Stock Price Models},
  series    = {Springer Finance},
  pages     = {243--272},
  publisher = {Springer},
  address   = {Berlin, Heidelberg},
  doi       = {10.1007/978-3-642-31214-4_9}
}

@article{RoperRutkowski2009,
  author  = {Roper, Michael and Rutkowski, Marek},
  year    = {2009},
  title   = {On the Relationship Between the Call Price Surface and the Implied Volatility Surface Close to Expiry},
  journal = {International Journal of Theoretical and Applied Finance},
  volume  = {12},
  number  = {4},
  pages   = {427--441},
  doi     = {10.1142/S0219024909005336}
}

@article{GaoLee2014,
  author  = {Gao, Kun and Lee, Roger W.},
  year    = {2014},
  title   = {Asymptotics of Implied Volatility to Arbitrary Order},
  journal = {Finance and Stochastics},
  volume  = {18},
  number  = {2},
  pages   = {349--392},
  doi     = {10.1007/s00780-013-0223-6}
}

@article{ChambersNawalkha2001,
  author  = {Chambers, Donald R. and Nawalkha, Sanjay K.},
  year    = {2001},
  title   = {An Improved Approach to Computing Implied Volatility},
  journal = {The Financial Review},
  volume  = {36},
  number  = {3},
  pages   = {89--100},
  doi     = {10.1111/j.1540-6288.2001.tb00021.x}
}

@article{ManasterKoehler1982,
  author  = {Manaster, Steven and Koehler, Gary},
  year    = {1982},
  title   = {The Calculation of Implied Variances from the Black--Scholes Model: A Note},
  journal = {The Journal of Finance},
  volume  = {37},
  number  = {1},
  pages   = {227--230},
  doi     = {10.1111/j.1540-6261.1982.tb01105.x}
}
\end{document}